\documentclass[11pt]{article}
\usepackage[letterpaper, margin=1in]{geometry} \usepackage{amsmath, amssymb, amsthm}
\usepackage{thmtools}
\usepackage{bbm}
\usepackage{graphicx}
\usepackage{url}
\usepackage{hyperref}
\usepackage[capitalise]{cleveref}
\usepackage{algorithm}
\usepackage[noend]{algpseudocode}
\usepackage{subcaption}
\usepackage{xspace}
\usepackage{mathrsfs}
\usepackage{xcolor}
\usepackage{tabularx}
\usepackage{booktabs}
\usepackage{tikz}
\usetikzlibrary{decorations.pathreplacing}

\newtheorem{theorem}{Theorem}[section]
\newtheorem{corollary}[theorem]{Corollary}
\newtheorem{lemma}[theorem]{Lemma}

\newtheorem{definition}[theorem]{Definition}
\newtheorem{claim}[theorem]{Claim}

\newenvironment{proof-sketch}{\noindent{\bf Sketch of Proof}\hspace*{1em}}{\qed\bigskip}
\newenvironment{proof-idea}{\noindent{\bf Proof Idea}\hspace*{1em}}{\qed\bigskip}
\newenvironment{proof-of-lemma}[1]{\noindent{\bf Proof of Lemma #1}\hspace*{1em}}{\qed\bigskip}
\newenvironment{proof-attempt}{\noindent{\bf Proof Attempt}\hspace*{1em}}{\qed\bigskip}

\makeatletter
\def\fnum@figure{{\bf Figure \thefigure}}
\def\fnum@table{{\bf Table \thetable}}
\long\def\@mycaption#1[#2]#3{\addcontentsline{\csname
  ext@#1\endcsname}{#1}{\protect\numberline{\csname
  the#1\endcsname}{\ignorespaces #2}}\par
  \begingroup
    \@parboxrestore
    \small
    \@makecaption{\csname fnum@#1\endcsname}{\ignorespaces #3}\par
  \endgroup}
\def\mycaption{\refstepcounter\@captype \@dblarg{\@mycaption\@captype}}
\makeatother

\newcommand{\mathify}[1]{\ifmmode{#1}\else\mbox{$#1$}\fi}
\newcommand{\bigO}O

\newcommand{\remove}[1]{}
\newcommand{\ignore}[1]{}

\def\R{\mathbb{R}}

 \usepackage{enumitem}
\hypersetup{
  colorlinks=true,
  linkcolor=black,
  filecolor=blue,
  urlcolor=blue,
  citecolor=blue,
  pdfpagemode=FullScreen,
}
\usepackage[sortcites, natbib=true, url=false, style=alphabetic, doi=false, isbn=false, maxbibnames=9, maxalphanames=4,minalphanames=3]{biblatex}
\AtEveryBibitem{\clearfield{note}
  \clearlist{language}
  \clearfield{location}
}
\newcommand{\E}{\mathbb E}
\newcommand{\Prb}{\mathbb P}
\newcommand{\ALG}{\mathrm{ALG}}
\newcommand{\OPT}{\mathrm{OPT}}

\newcommand{\DistProb}{\textsc{DistProb}\xspace}
\newcommand{\TimeDist}{\textsc{TimeDist}\xspace}
\newcommand{\DistCut}{\textsc{DistCut}\xspace}
\newcommand{\wopen}{f}

\usepackage{color-edits}
\addauthor[Shaofeng]{sj}{purple}

\title{The Power of Arrival Times in \\
Random-Order Online Facility Location}
\author{Yichen Huang\footnote{Supported by NSF grant CNS-2107078. Email: \texttt{yichenhuang@g.harvard.edu}.} \\Harvard University \and Shaofeng H.-C. Jiang\footnote{Email: \texttt{shaofeng.jiang@pku.edu.cn}.} \\Peking University}
\date{\today}

\begin{document}

\maketitle
\pagenumbering{gobble}
\begin{abstract}
We study online metric facility location with uniform opening costs in the random-order model (Meyerson FOCS'01).
The best previous upper bound was a $3$-competitive randomized algorithm (Kaplan, Naori, Raz SODA'23), leaving a gap to the best known lower bound of $2$.
In this work, we give two algorithms with improved competitive ratios:
(i) a deterministic algorithm with a competitive ratio below $2.42$ and (ii) a randomized algorithm with a competitive ratio below $2.59$ and the additional property that it retains the asymptotically optimal $O(\log n/\log \log n)$ competitive ratio in the adversarial-order model.
A key improvement is to take the arrival time of the request into consideration when making opening decisions: The arrival time carries geometric information about the local density around the request, which fundamentally helps the algorithm.

\end{abstract}
\newpage
\pagenumbering{arabic}

\section{Introduction}
Online facility location, introduced by Meyerson~\cite{Meyerson01online},
is a fundamental online optimization problem and has received significant attention in algorithmic study (e.g.,~\cite{Meyerson01online,GuoEtAl2020OnlineDynamic,Kaplan23Almost,Fotakis2011Survey,Fotakis2007PrimalDual,FotakisEtAl2025ImprovedPredictions,CyganCzumajMuchaSankowski2018,AnagnostopoulosBentUpfalVanHentenryck2004,ArgueEtAl2022LearningSample}).
In this problem, the algorithm is given in advance a metric space $(X,d)$ and a facility cost parameter $\wopen > 0$,
and it must serve a sequence of request points $(v_1, \ldots, v_n) \subseteq X$.
Upon receiving a request $v_i$, the algorithm may open new facilities at any point in the metric space,
paying $\wopen$ \emph{opening cost} for each opening,
and must irrevocably assign the request to an open facility, denoted as $y_i \in X$, paying the corresponding metric distance $d(v_i, y_i)$ as the \emph{connection cost}.
The objective is to minimize the total opening and connection cost.
Besides the standard adversarial-order setting~\cite{Fotakis2003CompetitiveRatio, Fotakis2007PrimalDual, Meyerson01online},
the problem has also been studied in the random-order model~\cite{Meyerson01online, Kaplan23Almost}, where the requests are chosen adversarially but the order in which the requests arrive is a uniformly random permutation.

We focus primarily on the random-order model.
The best known competitive ratio in the random-order model is $3$, achieved by an algorithm belonging to a so-called \(q\)-\DistProb family~\cite{Kaplan23Almost}, which includes Meyerson's original algorithm~\cite{Meyerson01online} whose ratio is $4$. 
In addition, this algorithm has a unique ``good-in-both-worlds'' feature:
the algorithm can simultaneously achieve an optimal $O(\log n / \log\log n)$ ratio in the adversarial-order setting.
Unfortunately, it is unclear if this ratio $3$ is tight in general for the random-order setting,  
since the current lower bound is only $2$~\cite{Kaplan23Almost}, which leaves a gap.
On the other hand, breaking $3$ requires overcoming a fundamental technical barrier,
since it has been shown that $3$ is the best ratio achievable for any algorithm in the $q$-\DistProb family~\cite{Kaplan23Almost}.

In this work, we break this \(3\)-competitive barrier
by devising new families of algorithms beyond $q$-\DistProb.
In $q$-\DistProb, for a parameter $q > 0$,
the algorithm maintains the set $F$ of currently open facilities,
and when a request $x$ arrives, the algorithm opens a facility at $x$
with probability
$\min \{q \cdot {d(x, F)}/{\wopen}, 1\}$,
after which $x$ is connected to the nearest facility in $F$.

\paragraph{Breaking the Barrier via Arrival Time.}
The \DistProb family makes its decisions based solely on the current distance. We additionally make use of the arrival time of each request.
Specifically, we consider a more general family of algorithms, called \TimeDist, in which the decision of whether to open a facility at request $v_t$ is determined by a function $g : (t, d(v_t, F)) \mapsto [0, 1]$ describing the probability of opening the request, i.e.,
\begin{equation}
    \label{eqn:distcut}
    \text{open facility at $v_t$ with probability } g(t, d(v_t, F)).
\end{equation}
We show in \Cref{thm:intro-distcut} that an optimal algorithm in this family achieves a fundamentally better ratio of $2.42$.

\begin{theorem}[See \Cref{thm:distcut}]
\label{thm:intro-distcut}
There is a deterministic algorithm for online facility location that is \(2.42\)-competitive in the random-order model.
\end{theorem}

Although the function $g$ can in principle take any real value in $[0, 1]$ (for example, $q$-\DistProb gives one realization of $g$), our $2.42$ ratio is achieved, remarkably, by a \emph{deterministic} algorithm whose image is the discrete set $\{0, 1\}$.
To the best of our knowledge, this is the first deterministic algorithm with a constant competitive ratio for online facility location in the random-order model.
Prior to this work, the best deterministic guarantee followed from the optimal adversarial-order ratio \(\Theta(\log n/\log\log n)\)~\cite{Fotakis2003CompetitiveRatio}.
We show in \cref{thm:time-distance-lower-bound} that this deterministic $g$ already gives the best ratio among all algorithms in the \TimeDist family.

\paragraph{Augmenting \DistProb with Time.} 
Our algorithm realizing \cref{thm:intro-distcut} relies heavily on random order and does not retain the adversarial-order guarantee. We next ask whether one can retain such a guarantee (as in~\cite{Meyerson01online,Kaplan23Almost}) while still breaking the ratio $3$. We give an affirmative answer by augmenting the \DistProb family with arrival-time information.

We propose a natural family called $q_t$-\DistProb, which generalizes $q$-\DistProb by replacing the fixed parameter $q$ with a time-dependent parameter $q_t$ in the opening probability for the $t$-th request $v_t$.
We also use the convention that the algorithm opens deterministically if $d(v_t, F) \ge \wopen$, i.e.,
\begin{equation}
    \label{eqn:qt_distprob}
    \Prb[v_t \text{ opens}] = \begin{cases}
        1 & d(v_t, F) \ge \wopen,\\
        q_t \cdot \frac{d(v_t, F)}{\wopen} & \text{otherwise}.
    \end{cases}
\end{equation}
Our second result, \Cref{thm:intro-general-q}, shows that there is an algorithm in the $q_t$-\DistProb family that achieves competitive ratio $2.59$, improving over $3$, while still maintaining the asymptotically optimal $O(\log n / \log\log n)$ ratio for adversarial-order input.

\begin{theorem}[See \Cref{thm:general-q}]
\label{thm:intro-general-q}
    There is an algorithm in the $q_t$-\DistProb family that is $2.59$-competitive in the random-order model and $O(\log n/\log \log n)$-competitive in the adversarial-order model.
\end{theorem}

Finally, we remark that the algorithms realizing \Cref{thm:intro-general-q,thm:intro-distcut} are quite simple.
Specifically, $g$ is an indicator for whether a linear inequality holds (see~\cref{eqn:distcut-final}), and $q_t$ is a two-phase piecewise-constant function for $q_t$-\DistProb (see~\cref{eqn:distprob-final}).
This simplicity leads to very efficient \emph{offline} algorithms.
In particular, the $2.42$ ratio in \Cref{thm:intro-distcut} can be readily implemented in $O(n^2)$ time in the offline setting, thereby giving an efficient randomized approximation algorithm for facility location.
While this is slightly worse than a $(1 + \sqrt{2} + \epsilon)$-approximation~\cite{Meyerson01online,CharikarG05}, which, to the best of our knowledge, is the best known ratio for $O(n^2)$-time algorithms\footnote{There exist better algorithms that run in $\widetilde O(n^2) = O(n^2 \mathrm{poly}\log(n))$ time (e.g., $(1.52+\varepsilon)$~\cite{MahdianYZ02}), but $(1+\sqrt 2 + \varepsilon)$ seems to be the best known ratio in strict $O(n^2)$ time.}, our algorithm is much simpler and easier to implement in practice.

\subsection{Our Techniques}
\label{subsec:overview}

The key improvement over the \DistProb family comes from utilizing \emph{the arrival times} of the requests. Indeed, the main challenge for an online algorithm is to distinguish between regions with many requests and regions with few requests, so that it tends to open facilities in the former. The arrival time in a random-order input contains exactly this information: For a region containing \(k\) requests, the first request from that region has expected arrival time roughly \(n/(k+1)\). Thus, dense regions tend to reveal themselves early. Conversely, if the first request from a region is observed only late in the sequence, this is evidence that the region is unlikely to contain many requests.

The use of arrival times is also necessary: we prove in \cref{thm:memoryless} that no algorithm can achieve a competitive ratio below \(3\) if it decides whether to open facilities at request locations $v_t$ using a time-oblivious rule $g(v_t, F_{t-1})$ of the current request $v_t$ and open facilities $F_{t-1}$. This includes the \(q\)-\DistProb family of~\cite{Kaplan23Almost}.

Both of our algorithms use the arrival time to calibrate their decisions, but they implement this principle through different mechanisms, yielding incomparable guarantees and requiring different analyses.
The corresponding sections are self-contained and may be read independently in either order.
Throughout the paper, we assume without loss of generality that the facility-opening cost is \(f=1\), by scaling all distances by \(1/f\).

\subsubsection{The Deterministic \TimeDist Algorithm}
The \TimeDist family allows the opening decision to depend on both the time and the distance. The particular algorithm realizing \cref{thm:intro-distcut} is called $\mu$-\DistCut. Parameterized by $\mu$, it uses the clock \(z_t := (t-1)/n\) and opens a facility in round $t$ at the request \(v_t\) if
\begin{equation}
\label{eqn:distcut-final}
    d(v_t,F_{t-1})\ge \min\{1, z_t/\mu\}, 
\end{equation}
and then connects \(v_t\) to its nearest open facility.

An important advantage of this deterministic rule is that it avoids the ``rent-then-buy'' pitfall: if many copies of the same location appear, a randomized \(\DistProb\)-type algorithm may fail to open several times, repeatedly paying connection cost before eventually buying. Our cutoff algorithm does not behave this way. For a fixed request location \(v\), the distance \(d(v,F_t)\) is nonincreasing over time, while the clock \(z_t\) is increasing. Hence, once \(v\) fails the cutoff, every later copy of \(v\) also fails it. The algorithm therefore makes an open-or-never-open decision at the first appearance, using the first arrival time as evidence about the density of the surrounding region.

\DistCut admits a remarkably simple analysis that anchors solely at the first-arriving request of a cluster, rather than at a carefully defined opening event as is common in analyses of randomized algorithms~\cite{Kaplan23Almost,Meyerson01online}. The deterministic rule certifies a close facility around the first request regardless of whether it opens, whereas in \DistProb-like randomized algorithms an unopened request provides no comparable guarantee.

Fix an offline cluster \(C\) served by an offline facility $c$, and let \(p=v_T\) be its first request to arrive. If \(p\) opens, then \(p\) itself is an anchor facility inside the cluster. If \(p\) does not open, then the failed cutoff condition certifies that \(d(p, F_{T-1}) \le z_T /\mu\). In either case, \(d(p,F_T)\le z_T/\mu\). We then bound the expected cost of every future request \(u\) by \(d(p, F_T) + (1+\mu) d(p, u)\) (\cref{lem:anchor}). Thus, the total cost contributed by \(C\) is bounded by
\[
\E \left[\mathbbm 1_{p \text{ opens}} + \sum_{u \in C} (d(p, F_T) + (1+\mu) d(p, u))\right].
\]

By the triangle inequality, \(d(p,u) \le d(u, c)+d(p, c)\). Since \(p\) is uniformly distributed over \(C\), the expected value of \(d(p, c)\) is the average offline connection cost in $C$, and thus $\E[\sum_{u \in C} (1+\mu) d(p, u)] \le 2(1+\mu)R(C)$, where $R(C) = \sum_{u \in C} d(u, c)$ is the offline connection cost.

The remaining terms are $\sum_{u \in C} d(p, F_T) = |C| d(p, F_T)$ and the potential opening cost at $p$. Importantly, random order creates opposing forces between \(|C|\) and \(z_T\): if \(|C|\) is large, then the first request of \(C\) is likely to arrive early, when \(z_T\) is small. In particular, for every \(\ell\),
\[
\Prb[z_T \ge z_\ell] = \Prb[T \ge \ell] \le \left(\frac{n-\ell+1}{n}\right)^{|C|} = (1-z_\ell)^{|C|}.
\]
\cref{lem:core} formalizes this tradeoff and shows $\E[\mathbbm 1_{p \text{ opens}} + |C|d(p, F_T)] \le 1+e^{-(1+\mu)}/\mu$. Recall that the $d(p, u)$ term contributes \(2(1+\mu)R(C)\). Since the offline solution pays $1+R(C)$ for this cluster, the competitive ratio is upper bounded by $\max\{1+\frac{e^{-(1+\mu)}}{\mu}, 2(1+\mu)\}$, which balances at $\mu^* \approx 0.21$ and gives a competitive ratio below \(2.42\).

We remark that this algorithm relies crucially on the random-order assumption and is not robust to adversarial order. In particular, we show an $\Omega(\sqrt n)$ lower bound on the competitive ratio of this algorithm in Appendix~\ref{appendix:cut-adv}.

\subsubsection{$q_t$-\DistProb}
Next, we aim to retain the adversarial-order guarantee while still outperforming the ratio \(3\). A natural approach is to incorporate time into the \DistProb family by making the multiplier \(q_t\) time-dependent.
Our algorithm is parametrized by a nonincreasing sequence \(q=(q_1,\ldots,q_n)\in[0,1]^n\).
When the request at time \(t\) is at distance \(d\) from the current facility set, the algorithm opens it with probability \(g_t(d)\), where
\[
g_t(d):= \begin{cases}
    1 & d \ge 1,\\
    q_t d & d < 1.
\end{cases}
\]
The deterministic cutoff at distance \(1\) is important: unlike the truncation rule \(g(d)=\min\{q d,1\}\) analyzed in~\cite{Kaplan23Almost,Meyerson01online}, the cutoff prevents the algorithm from paying too much connection cost during late rounds when \(q_t\) may be very small.

\paragraph{Selecting the sequence.}
For any nonincreasing \(q\), we prove a competitive ratio of \(\max\{2(1+\bar q),1+\rho(q)\}\), where \(\bar q := \sum_{t=1}^n q_t/n\) is the average of the \(q_t\)'s and \(\rho(q)\) is a function that decreases coordinate-wise in \(q\) and roughly captures the connection cost of ``not opening'' (see \cref{thm:general-q}).
Optimizing this bound gives a front-loaded threshold rule \(q_t=\mathbbm 1_{t\le \alpha^* n}\), where \(\alpha^*\approx 0.293\), which gives a ratio of about \(2.586\).
This sequence alone is not robust to adversarial order, since it sets \(q_t=0\) for a constant fraction of the sequence.
We therefore add a small constant \(\varepsilon\) and use
\begin{equation}
\label{eqn:distprob-final}
    q_t := \begin{cases}
        1 & t \le \alpha^* n,\\
        \varepsilon & t > \alpha^* n.
    \end{cases}
\end{equation}
This increases the random-order ratio by only \(O(\varepsilon)\), while restoring the \(O_\varepsilon(\log n/\log\log n)\) adversarial-order guarantee.

\paragraph{Analysis: Decomposing the Cost.}

\begin{table}[t]
\centering
\small
\renewcommand{\arraystretch}{1.15}
\begin{tabularx}{\textwidth}{@{}p{0.27\textwidth}X@{}}
\toprule
\textbf{Definition} & \textbf{Meaning} \\
\midrule
$C, c$ & Offline cluster, offline facility serving the cluster\\
$r_u = d(c, u)$ & Offline connection cost\\
$d_t(u)=d(u,F_{t-1})$ & True distance from request $u$ to the online facility set before round $t$. \\
$D_t=d(c,F_{t-1})$ & Distance from the offline center $c$ to the online facility set before round $t$. \\
$a_t(u)=(D_t-r_u)^+$ & Center excess of request $u$ before round $t$. \\
$g_t$ & Function determining the opening probability at time $t$. \\
$g_t(d_t(u))$ & Probability that request $u$ opens at time $t$. \\
$g_t(a_t(u))$ & Probability that request $u$ has a balanced opening at time $t$. \\
$T$ & First time at which a balanced opening occurs in cluster $C$. \\
$C_{\le T}, C_{>T}$ & Requests of $C$ arriving up to and including time $T$, and after time $T$. \\
$H_C$ & Center-excess waiting term before the first balanced opening. \\
\bottomrule
\end{tabularx}
\caption{Notation for the analysis of $q_t$-\DistProb.}
\label{tab:qt-distprob-notation}
\end{table}

The analysis of \(q_t\)-\DistProb is relatively technical. We gradually introduce the notation we need, and summarize the notations in \cref{tab:qt-distprob-notation}.
As in the framework of~\cite{Kaplan23Almost}, we fix an offline optimum, and analyze each cluster \(C\) of requests in this optimum that are served by an offline facility $c$.
The analysis will use the first random ``balanced'' opening, which is a facility opened at request \(v_T\) at time \(T\), as a pivot.
The pivot \(v_T\) serves as a certificate for bounding the cost of requests after round \(T\), which we call the suffix.
Once the balanced pivot \(v_T\) opens, every later request \(v_t\in C\) can connect through \(v_T\), so its expected cost is at most \((1+q_t)d(v_t,v_T)\le (1+q_t)(r_{v_t}+r_{v_T})\), where \(r_v = d(v,c)\) denotes the offline connection cost of \(v\).
Hence, we would like \(v_T\) to be biased toward points close to \(c\), and this relies on the choice of balanced opening.\footnote{
If one were to simply use the first actual opening in \(C\), then farther requests are more likely to open, so such a pivot may be biased away from the offline center and would not serve as a good certificate for suffix costs.}

Now, we define the balanced opening and discuss the analysis. Our definition is different from that of~\cite{Kaplan23Almost} and is important for working with a varying multiplier \(q_t\).
A central quantity in defining the balanced opening is the \emph{center excess}, which measures the part of a request's distance to the online facility that is not already explained by its offline connection cost.
Let \(D_t=d(c,F_{t-1})\) be the distance from the offline center to the online facility set, and let \(r_u=d(u,c)\) be the offline connection cost of \(u\).
For an unarrived request \(u\in C\), define its center excess by \(a_t(u)=(D_t-r_u)^+\). Note that $d_t(u) \in [a_t(u), a_t(u)+2r_u]$.

The algorithm opens \(u\) with probability \(g_t(d_t(u))\), where \(d_t(u) = d(u, F_{t-1})\) is the true distance, but we declare this opening balanced only with probability \(g_t(a_t(u))\).
The probabilities are coupled so that balanced openings are always actual openings.\footnote{More precisely, we sample a coin \(x \sim \mathrm{Unif}[0, 1]\). If \(x \le g_t(a_t(u))\), then we open a facility at \(u\) and this opening is balanced; if \(x \in (g_t(a_t(u)), g_t(d_t(u))]\), then we open a facility but the opening is not balanced; if \(x > g_t(d_t(u))\), then we do not open a facility there.}
Equivalently, a balanced opening is the hypothetical opening that would occur if the algorithm used \(a_t(u)\) in place of the true distance \(d_t(u)\).
The time \(T\) is defined as the first time when a balanced opening occurs, and we write \(C_{\le T}\) for the requests up to and including \(T\) and \(C_{> T}\) for the requests after \(T\).

This choice of balance has two useful properties.
First,  $a_t(u) = (D_t - r_u)^+$ is larger for requests closer to \(c\), so the balanced pivot is biased toward smaller offline radii, which is exactly what is needed to serve the suffix of the cluster.
Second, for requests before round \(T\), we would like to analyze the process as if the true distance were \(a_t\). In this hypothetical process, the prefix stops at the first balanced opening, and the expected cost is easily bounded by \(1+H_C\), where \(1\) is the opening cost and
\[
H_C := \E\!\left[
\sum_{v_t\in C_{\le T}}
(1-g_t(a_t(v_t)))a_t(v_t)
\right]
\]
is the center-excess waiting cost.
However, the true distance is \(d_t\), not \(a_t\), so the algorithm may pay additional connection cost and may also open facilities that are declared imbalanced.
Importantly, the difference between the distance used for opening, \(d_t\), and the distance used for analysis, \(a_t\), satisfies \(d_t(u) - a_t(u) \le 2r_u\).
Thus, these additional costs can be charged directly to the offline connection cost.
Overall, we obtain in \cref{clm:decomposition} the following decomposition of the algorithm's expected cost on the requests in \(C\):
\[
\E[\ALG(C)] \le 1 + H_C + 2(1+\bar q)R(C),
\]
where \(R(C)=\sum_{u\in C}r_u\).
For a fixed multiplier \(q_t\equiv q\), the term \(H_C\) is easy to control by \(H_C\le 1/q\), recovering the sharp fixed-\(q\) bound of~\cite{Kaplan23Almost}.
For a varying sequence \(q_t\), however, controlling \(H_C\) is the main remaining challenge.

\paragraph{Controlling a Varying Sequence.}
The decomposition isolates the only term that truly depends on the exact evolution of \(q_t\), rather than just on the average \(\bar q\).
For a general varying sequence \(q_t\), the difficulty is that \(H_C\) depends on the adaptive evolution of the configuration.
After a request fails to become a balanced opening, the algorithm may still open it as an imbalanced facility; this can decrease the distance \(D_t\) from the offline center to the online facility set, and hence decrease the future excesses \(a_t(u)\).
Rather than tracking this geometry exactly, we dominate it by a geometry-free process.

In this relaxed process, each remaining request \(u\) at time \(t\) carries an excess value \(a_t(u)\in[0,1]\).\footnote{We cap the excess at \(1\), because if \(a_t(u) \ge 1\), then a balanced opening happens with probability \(1\), which is equivalent for this process to setting \(a_t(u)=1\).}
A uniformly random remaining request arrives.
With probability \(q_t a_t(u)\), a balanced opening occurs and the process stops.
Otherwise, no balanced opening occurs: the process pays \(a_t(u)\), removes that request, and then an adversary may decrease the remaining excesses arbitrarily, so each future value \(a_{t+1}(v)\) for $v \ne u$ can be any value at most \(a_t(v)\).
This process relaxes the original dynamics by forgetting the metric structure and allowing the adversary to choose the most expensive possible decreases in the remaining excesses.
Therefore, the cost paid in this process upper-bounds \(H_C\) (\cref{lem:relaxation}).

Next, we establish the equalization lemma (\cref{clm:adaptive-equalization}), which shows that despite this adversarial freedom, the worst profile is completely symmetric: all remaining requests have the same excess, and this common excess stays fixed throughout the process.
This reduces the adaptive metric problem to the scalar function
\[
\rho(q) := \sup_{x \in [0, 1]} \left \{x\sum_{t=1}^n \prod_{s \le t}(1-q_sx) \right\}
\]
appearing in \cref{thm:general-q}, which we then optimize over.
It is important here that the center excesses are monotone: as facilities are added, the values \(a_t(u)\) can only decrease.
By contrast, the balanced-opening rule of~\cite{Kaplan23Almost}, which is based on prefix minima over the remaining requests, does not have this monotonicity. Our center-excess definition is therefore crucial for making the time-varying analysis possible.

\subsection{Related Work}
Online facility location has been studied in many different settings. In the adversarial-order model, \cite{Fotakis2003CompetitiveRatio} established the asymptotically tight competitive ratio \(\Theta(\log n/\log\log n)\), observed that \DistProb attains the corresponding upper bound, and also gave a matching deterministic algorithm. Subsequent work developed a simpler deterministic algorithm~\cite{AnagnostopoulosBentUpfalVanHentenryck2004} and a primal-dual algorithm for nonuniform facility-opening costs~\cite{Fotakis2007PrimalDual}; see the survey of Fotakis~\cite{Fotakis2011Survey}. The variant with a prescribed set of candidate facilities is studied in \cite{LiEtAl2026DiscreteRandomOrder}, which obtained an \(8\)-competitive algorithm in the random-order model. Their model is more general, but their guarantee is correspondingly weaker. Beyond the random-order model, online facility location has also been studied in other beyond-worst-case settings, including \(t\)-bounded semirandom adversaries~\cite{Lang2018TBoundedAdversary}, partially adversarial arrival orders~\cite{Kaplan23Almost}, online-with-a-sample settings~\cite{ArgueEtAl2022LearningSample}, and learning-augmented settings~\cite{JiangEtAl2022Predictions,AlmanzaEtAl2021MultipleAdvice,FotakisEtAl2025ImprovedPredictions,AzarPanigrahiTouitou2022}. Other variants include facility leasing~\cite{NagarajanWilliamson2013FacilityLeasing}, departures and deletions~\cite{CyganCzumajMuchaSankowski2018}, and online and dynamic models with recourse~\cite{GuoEtAl2020OnlineDynamic}.

The random-order model traces its roots to the classical secretary problem,
a canonical problem in the theory of optimal stopping
\cite{ChowMorigutiRobbinsSamuels1964,Ferguson1989Secretary,
ChowRobbinsSiegmund1971}, and constitutes an important paradigm in beyond-worst-case analysis~\cite{Roughgarden2021BeyondWorstCase}. The model has been studied for a wide range of online problems, including set cover~\cite{GuptaKehneLevin2022RandomOrderSetCover},
edge coloring~\cite{BahmaniMehtaMotwani2012RandomOrderEdgeColoring,
BhattacharyaGrandoniWajc2021Nibble}, weighted bipartite
matching~\cite{KesselheimEtAl2013WeightedMatching}, packing linear
programs~\cite{KesselheimEtAl2018PackingLP}, and knapsack and generalized
assignment~\cite{AlbersKhanLadewig2021KnapsackGAP}.  We refer the reader
to~\cite{GuptaSingla2021RandomOrder} for a survey of random-order models.

\section{Preliminaries}
\label{sec:preliminaries}
For a positive integer \(m\), let \([m]:=\{1,\ldots,m\}\). For a set $S$, write $\mathrm{Unif}(S)$ for the uniform distribution over $S$. For a number $x \in \R$, we write $x^+ := \max\{x, 0\}$. On a metric space $(M, d)$, for a set $F \subseteq M$, write $d(u, F):= \min_{v \in F} d(u, v)$ with the convention that $d(u, \varnothing) = \infty$.

In the online (uncapacitated, metric, uniform) facility location problem, the algorithm is given a metric space \((M,d)\). Over time, requests $v_1, \ldots, v_n$ arrive sequentially. We assume that $n$ is known to the algorithm. After observing \(v_t\), an online algorithm may irrevocably open facilities at points of \(M\), paying $1$ for each newly opened facility.\footnote{This is without loss of generality by scaling all distances by the original opening
cost.} It must then irrevocably assign \(v_t\) to an open facility, which we may assume is the nearest open facility.
The total online cost is
\[
    \ALG= |F_n|+\sum_{t=1}^n d(v_t,F_t).
\]

Let $U := \{v_1, \ldots, v_n\}$ be the request multiset. An optimal offline solution knows the request multiset $U$ in advance and chooses a facility set \(F\subseteq M\) and assigns each request \(u\in U\) to the nearest facility. Its cost is
\[
    \OPT = \min_{F \subseteq M} \{|F|+\sum_{u\in U}d(u,F)\}.
\]

The performance of an online algorithm is measured by its competitive ratio. Here, we consider the random-order model, where the adversary chooses the request multiset $U$, after which the request copies are presented in a uniformly random order. A (deterministic or randomized) online algorithm has competitive ratio $\lambda$ if for every request multiset $U$, $\E[\ALG] \le \lambda \cdot \OPT$, where the expectation is taken over the random permutation for deterministic algorithms, and also over the algorithm’s internal randomness for randomized algorithms. Note that the optimal solution does not depend on the random order of the requests and is hence deterministic.

 \section{Time-Based Cutoff Algorithm}
\label{sec:deterministic}
In this section, we analyze the deterministic cutoff algorithm described in~\cref{alg:clocked-cutoff} and prove \cref{thm:intro-distcut}.

\begin{algorithm}[h]
\caption{$\mu$-\DistCut}
\label{alg:clocked-cutoff}
\begin{algorithmic}[1]
\State \(F_0 \gets \varnothing\)
\For{\(t=1,\ldots,n\)}
    \State Receive request \(v_t\)
    \State \(z_t \gets (t-1)/n\)
    \If{\(d(v_t, F_{t-1}) \ge \min\{1,z_t/\mu\}\)}
\State \(F_t \gets F_{t-1}\cup\{v_t\}\)
    \Else
        \State \(F_t \gets F_{t-1}\)
    \EndIf
    \State Assign \(v_t\) to its nearest facility in \(F_t\)
\EndFor
\end{algorithmic}
\end{algorithm}

\begin{theorem}
\label{thm:distcut}
For every $\mu \in (0, 1]$, the algorithm $\mu$-\DistCut satisfies, for every request multiset,
\[
\E[\ALG]
\le
\max\left\{
1 + \frac{e^{-(1+\mu)}}\mu,\,2(1+\mu)
\right\}\OPT.
\]
\end{theorem}

\begin{proof}
For a request \(u=v_t\), let \(\ALG(u)\) denote the service cost paid by
the algorithm in that round: \(\ALG(u)=1\) if the algorithm opens at \(u\), and
\(\ALG(u)=d(u,F_{t-1})\) otherwise. For a set \(S\) of request copies, write
\(\ALG(S):=\sum_{u\in S}\ALG(u)\).

Fix an optimal solution. We analyze the algorithm separately on each optimal
cluster. Let $c$ be an offline facility and $C$ be all the requests served by $c$. Write
\(k=|C|\), \(r_u:=d(u,c)\), and \(R(C):=\sum_{u\in C}r_u\). Thus
\(\OPT(C)=1+R(C)\).

The proof uses the first request of \(C\) as an anchor. We first show that the cost of any later request in \(C\) can be bounded by the distance from this anchor to the current facility set, plus the distance from the later request to the anchor.

\begin{lemma}[First-arrival anchor]
\label{lem:anchor}
Fix a deterministic prefix ending at round \(T\), such that \(p=v_T\) is the first request from $C$ to appear. Conditioned on this history, for every
\(u\in C\setminus\{p\}\),
\[
\E[\ALG(u)]\le d(p,F_T)+(1+\mu)d(p,u),
\]
where the expectation is only over the random completion after round \(T\).
\end{lemma}

\begin{proof}
Fix \(u\in C\setminus\{p\}\), and let \(j\) be its arrival time in the random
completion. Then \(j\) is uniform over \(\{T+1,\ldots,n\}\). When \(u\) arrives, it has an available connection of length at most \(D=d(p,F_T)+d(p,u)\). If \(D\ge 1\), then \(\ALG(u)\le 1\le D\le d(p,F_T)+(1+\mu)d(p,u)\). It remains to consider $D < 1$.

We have
\[
\ALG(u)\le D+(1-D)\mathbbm 1_{\{z_j\le \mu D\}}.
\]
If no future time satisfies \(z_j\le \mu D\), then \(\E[\ALG(u)]\le D\), and
we are done. Otherwise \(z_{T+1}\le \mu D\). Since \(\mu\le1\), this implies
\(z_{T+1}\le D\), and hence \(1-D\le 1-z_T-1/n\).

Among the \(n-T\) future times, at most \(n(\mu D-z_T)\) satisfy
\(z_j\le \mu D\). The algorithm pays $1$ at those times and pays $D$ otherwise.
Hence
\[
\E[\ALG(u)]
\le D+(1-D)\frac{n(\mu D-z_T)}{n-T} = D+(1-D)\frac{\mu D-z_T}{1-z_T - \frac{1}n}
\le D+\mu D-z_T.
\]
Observe that \(z_T\ge \mu d(p,F_T)\): either $z_T \ge \mu d(p, F_{T-1}) = \mu d(p, F_T)$, or $p$ opens and $d(p, F_T) = 0$. Therefore, the last expression is at most
\[
D+\mu\bigl(D-d(p,F_T)\bigr)
=
d(p,F_T)+(1+\mu)d(p,u). \qedhere
\]
\end{proof}

Let \(p=v_T\) denote the random first request of \(C\) in the actual random order. We define the following first-request charge of $C$:
\[
\Gamma_C:= \mathbbm 1_{\{p\text{ opens}\}} + k\cdot d(p, F_T).
\]

\begin{lemma}[Cluster decomposition]
\label{lem:decomposition}
For every optimal cluster \(C\),
\[
\E[\ALG(C)]\le \E[\Gamma_C]+2(1+\mu)R(C).
\]
\end{lemma}

\begin{proof}
Consider any prefix that ends with the first request in \(C\), and denote by $p$ this first request and by $T$ its arrival time. For request $p$, the algorithm pays an opening cost of $1$ if it opens and a connection cost of $d(p, F_T)$. Conditioned on this prefix, the remaining suffix is uniformly random. Applying Lemma~\ref{lem:anchor} to every other request $u \in C\setminus \{p\}$ gives that, for every such prefix, the expected cost of \(C\) is at most
\begin{align*}
    &\mathbbm 1_{\{p\text{ opens}\}} + d(p, F_T)+ \sum_{u\in C\setminus\{p\}}(d(p,F_T) + (1+\mu)d(p, u)) \\ = \; & \mathbbm 1_{\{p\text{ opens}\}} + k\cdot d(p, F_T) + (1+\mu)\sum_{u\in C\setminus\{p\}}d(p,u).
\end{align*}

Averaging over all possible first-\(C\) prefixes gives
\[
\E[\ALG(C)]
\le \E[\Gamma_C]
+(1+\mu)\E\!\left[\sum_{u\in C\setminus\{p\}}d(p,u)\right].
\]

It remains to bound the last expectation. Since the arrival order is uniform,
the first request \(p\) of \(C\) is a uniformly random request in \(C\).
For every fixed \(p\), the triangle inequality gives
\[
\sum_{u\in C\setminus\{p\}}d(p,u)
\le \sum_{u\in C\setminus\{p\}}(r_p+r_u)
=R(C)+(k-2)r_p.
\]
Averaging over the uniform choice of \(p\), this is at most
\(R(C)+(k-2)R(C)/k\le 2R(C)\).
\end{proof}

It remains to bound \(\E[\Gamma_C]\), the expected cost of obtaining the first
anchor for the cluster. The following lemma formalizes the tradeoff that large clusters tend to have small first-arrival times.

\begin{lemma}[First-anchor bound]
\label{lem:core}
For every optimal cluster \(C\), \[\E[\Gamma_C]\le 1 + \frac{e^{-(1+\mu)}}\mu.\]
\end{lemma}

\begin{proof}
Fix the identity of the first request \(p\) of \(C\) and the relative order of all requests outside \(C\). Conditioned on this information, the set of times occupied by the \(k\) requests of \(C\) is still a uniformly random \(k\)-subset of \([n]\), and \(p\) occupies the minimum of these times.

For each possible first time \(t\) of \(C\), let \(d_t\) be the distance from
\(p\) to the open facility set immediately before serving \(p\) in the coupled
run in which the first \(t-1\) positions are filled by the first \(t-1\)
outside requests in the exposed outside order. As \(t\) increases, this prefix
only gains additional outside requests, and facilities are never closed. Hence
\(d_t\) is nonincreasing in \(t\), while \(z_t\) is nondecreasing. Therefore, once \(p\) fails to open at some possible time, it also fails to open at every later possible time.

If \(p\) opens at every possible time, then \(\Gamma_C=1\), and there is
nothing to prove. Otherwise, let \(\ell\) be the first possible time at which
\(p\) does not open. Failure to open at time \(\ell\)
means \(d_\ell<1\) and \(z_\ell>\mu d_\ell\).

If \(T<\ell\), then \(p\) opens and \(\Gamma_C=1\). If \(T\ge \ell\), then
\(p\) fails to open and \(\Gamma_C=kd_T\le k\min\{z_\ell/\mu, 1\}\). If \(k\min\{z_\ell/\mu,1\}<1\), this is at most \(1\). We may therefore assume \(k\min\{z_\ell/\mu,1\}\ge1\), in which case the conditional expected charge is at most
\[
1+\bigl(k\min\{z_\ell/\mu, 1\}-1\bigr)\Pr[T\ge \ell].
\]

The event \(T\ge \ell\) is exactly the event that all \(k\) times occupied by
\(C\) lie in \(\{\ell,\ldots,n\}\). Thus,
\[
\Pr[T\ge \ell]
=
\frac{\binom{n-\ell+1}{k}}{\binom nk}
\le
\left(\frac{n-\ell+1}{n}\right)^k
=(1-z_\ell)^k.
\]

Thus, the conditional expected charge is at most
\[
1+(k \min\{z_\ell/\mu, 1\}-1)(1-z_\ell)^k.
\]

The remaining step is purely algebraic: we upper-bound the above expression for every value of $z_\ell$ and $k$. First suppose \(z_\ell\le \mu\). Then $\min\{z_\ell/\mu, 1\} = z_\ell / \mu$ and
\[
1 + (kz_\ell/\mu -1)(1-z_\ell)^k \le 1+(kz_\ell/\mu-1)e^{-k z_\ell},
\]
which is maximized at $k z_\ell = 1 + \mu$, where its value is \(1+e^{-(1+\mu)}/\mu\).

Now suppose \(z_\ell \ge \mu\). Then \(\min\{z_\ell/\mu, 1\}=1\), and the cost is upper-bounded by
\begin{align*}
    1+(k-1)(1-\mu)^k&\le 1+(k-1)e^{-k \mu}\\&=1+\frac{e^{-\mu}}{\mu}\cdot \mu(k-1)e^{-\mu(k-1)}
\\&\le1+
\frac{e^{-(1+\mu)}}{\mu},
\end{align*}
using $1-x \le e^{-x}$ and \(xe^{-x}\le e^{-1}\).

The bound holds for every identity of \(p\) and every outside relative order, so taking expectation over these realizations proves the lemma.
\end{proof}

By Lemmas~\ref{lem:decomposition} and~\ref{lem:core}, every optimal cluster
\(C\) satisfies
\[
\E[\ALG(C)]\le 1 + \frac{e^{-(1+\mu)}}\mu+2(1+\mu)R(C).
\]
Since \(\OPT(C)=1+R(C)\), summing over all clusters of the fixed optimal solution gives the theorem.
\end{proof}

\cref{thm:intro-distcut} then follows from taking $\mu^* < 0.21$ to be the solution to $1 + \frac{e^{-(1+\mu)}}\mu=2(1+\mu)$.

A natural question is whether a different, possibly nonlinear, threshold in terms of $z_t$ might achieve a better result. Complementing our upper bound, we show in \cref{thm:time-distance-lower-bound} that \DistCut is optimal among all algorithms that open at request locations using only time and distances, even randomized ones. Thus, any future improvement must either use a richer form of memory or open facilities at ambient locations.

As noted in the Introduction, the algorithm crucially relies on the random-order property, at the cost of sacrificing the adversarial-order guarantee. In particular, we show in Appendix~\ref{appendix:cut-adv} that it has competitive ratio $\Omega(\sqrt n)$ in adversarial order. \section{Time-varying \DistProb}
\label{sec:randomized}
In this section, we study the time-varying \DistProb as defined in \cref{alg:distprob}. When observing request $v_t$ in round $t$, the algorithm opens deterministically if $d(v_t, F_{t-1}) \ge 1$, and otherwise opens with probability $q_t \cdot d(v_t, F_{t-1})$. 

\begin{algorithm}[h]
\caption{$q_t$-\DistProb}
\label{alg:distprob}
\begin{algorithmic}[1]
\State \(F_0 \gets \varnothing\)
\For{\(t=1,\ldots,n\)}
    \State Receive request \(v_t\)
    \State Set $g_t(d) := \begin{cases}
        1 & d \ge 1,\\
        q_t d& d < 1.
    \end{cases}$
    \State With probability $g_t(d(v_t, F_{t-1}))$, \(F_t \gets F_{t-1}\cup\{v_t\}\)
    \State Otherwise, \(F_t \gets F_{t-1}\)
    \State Assign \(v_t\) to its nearest facility in \(F_t\)
\EndFor
\end{algorithmic}
\end{algorithm}

We will prove \cref{thm:intro-general-q} via a chain of arguments. First, we provide a general analysis of $q_t$-\DistProb.

\begin{theorem}
\label{thm:general-q}
    For every $\{q_t\}$ such that $q_t \in [0, 1]$ and $q_t \ge q_{t+1}$, the algorithm $q_t$-\DistProb satisfies
    \[\E[\ALG] \le \max\{1+\rho(q), 2(1+\bar q)\}\OPT,\]
    where
    \[\bar q := \frac 1 n  \sum_{i=1}^n q_i, \quad \rho(q) := \sup_{x \in [0, 1]} \left \{x\sum_{t=1}^n \prod_{s \le t}(1-q_sx) \right\}.\]
\end{theorem}

The term \(1+\rho(q)\) controls the waiting cost before a balanced opening, while \(2(1+\bar q)\) controls the radius charges. The expression of $\rho(q)$ is somewhat complicated. Nevertheless, \cref{thm:general-q} reduces the problem to a non-metric, purely algebraic optimization problem. Next, we show that the optimal sequence for this upper bound is achieved by a two-phase sequence. 

\begin{corollary}[Benchmark-Optimal Sequence]
\label{thm:sequence}
    The sequence $q^*_t = \mathbbm 1_{\{t \le \alpha^* n\}}$ achieves competitive ratio $2(1+\alpha^*) + O(1/n)$ where $\alpha^* \approx 0.293$ is the unique solution to 
    \[
    \frac{1-\alpha}{\alpha} e^{-1/(1-\alpha)}= 2\alpha.
    \]
\end{corollary}

Finally, to establish the adversarial-order guarantee, we show that any sequence $q$ bounded above and below by an absolute constant retains the same asymptotic adversarial-order guarantee.

\begin{lemma}[Adversarial Robustness]
\label{thm:bounded-clock-adversarial}
For every pair of constants $0 < q_- \le q_+ \le 1$ and every $\{q_t\}$ such that \(q_-\le q_t\le {q_+}\) for every round \(t\), the algorithm $q_t$-\DistProb is $O_{q_-, q_+}\!\left(\frac{\log n}{\log\log n}\right)$-competitive under adversarial arrival order.
\end{lemma}

Before delving into the proofs, we show how to assemble these theorems to get \cref{thm:intro-general-q}.
\begin{proof}[Proof of \cref{thm:intro-general-q}]
    Note that we cannot apply \cref{thm:bounded-clock-adversarial} to our sequence $q^*_t = \mathbbm 1_{\{t \le \alpha^* n\}}$ defined in \cref{thm:sequence} since $q^*_t$ is not bounded away from $0$. Nevertheless, we can take arbitrary small constant $\varepsilon$ and set 
\[
q^{\varepsilon}_t := \mathbbm 1_{\{t \le \alpha^* n\}} + \varepsilon \mathbbm 1_{\{t > \alpha^* n\}}.
\]
Clearly, $\bar q^\varepsilon \le \bar q^{*} + \varepsilon$. On the other hand, since $\rho(q)$ is decreasing in every index of $q$, we have $\rho(q^\varepsilon) \le \rho(q^*)$. Thus, 
\[
\max\{1+\rho(q^\varepsilon), 2(1+\bar q^\varepsilon)\} \le \max\{1+\rho(q^*), 2(1+\bar q^* + \varepsilon)\} \le 2(1+\alpha^*+\varepsilon).
\]

Applying \cref{thm:bounded-clock-adversarial} to $q^\varepsilon$, we have that $q^\varepsilon_t$-\DistProb is $O_\varepsilon(\log n/\log \log n)$ competitive in the adversarial order and $2(1+ \alpha^* + \varepsilon)$ in the random order model. Taking $\varepsilon < 0.295- \alpha^*$ proves \cref{thm:intro-general-q}.
\end{proof}

The rest of this section is organized as follows. We first prove \cref{thm:general-q} following the framework described in~\cref{subsec:overview}: In Subsection~\ref{subsec:cost-decomposition}, we derive the cost decomposition and present its consequence in the fixed-$q$ case; then, in \cref{subsec:relaxation} we define and analyze the relaxed process for bounding the varying $q_t$. After that, \cref{subsec:optimize} establishes the optimality of two-phase sequences and proves \cref{thm:sequence}; \cref{subsec:robust} proves \cref{thm:bounded-clock-adversarial}.

\subsection{Cost Decomposition}
\label{subsec:cost-decomposition}
In this section, we prove the decomposition lemma. Most notations used in this subsection are summarized in \cref{tab:qt-distprob-notation}.

Fix a cluster \(C\) of an optimal solution with center \(c\), and write
\[
r_u:=d(u,c),\qquad R(C):=\sum_{u\in C}r_u,
\qquad D_t:=d(F_{t-1},c).
\]
For each round \(t\), recall \(g_t(x)\) denote the probability of opening at a request whose distance to the current facility set is \(x\), and define the expected one-step cost \(h_t(x):=g_t(x)+(1-g_t(x))x\), with the convention \(h_t(\infty)=1\). Note that \(h_t\) is nondecreasing, \((1+q_t)\)-Lipschitz, and satisfies \(h_t(x)\le(1+q_t)x\).

For a copy \(u\in C\) that has not yet arrived, define its \emph{center excess} \(a_t(u)=(D_t-r_u)_+\), which is the part of the current connection distance not already explained by the radius \(r_u\) paid by the optimal cluster. For \(t\ge2\), the triangle inequality gives
\[
a_t(u)\le d(u,F_{t-1})\le a_t(u)+2r_u.
\]
At \(t=1\), we have \(D_1=\infty\), and the above conventions give \(g_1(a_1(u))=1\).

Write $d_t(u) = d(u, F_{t-1})$. Recall that a facility is opened at request \(u\) in round \(t\) with probability \(g_t(d_t(u))\). We will create a balanced opening at $u$ with probability $\beta_t(u):=g_t(a_t(u))$, coupled with actual opening. More precisely, we sample \(\xi\sim \mathrm{Unif}[0,1]\). If \(\xi\le g_t(a_t(u))\), we open at $u$ and the opening is balanced; if \(\xi\in(g_t(a_t(u)), g_t(d_t(u))]\), it is an imbalanced opening; if $\xi > g_t(d_t(u))$, we do not open at $u$. Since $a_t(u) \le d(u, F_{t-1})$ and $g_t$ is nondecreasing, the definition is valid.

Let \(T\) be the first balanced opening from \(C\), with \(T=n+1\) if it never occurs. Define
\[
C_{\le T}=\{v_t\in C:t\le T\},
\qquad
C_{>T}=\{v_t\in C:t>T\}.
\]
Thus, \(C_{\le T}\) contains the pivot \(v_T\) whenever \(T\le n\). If \(T=n+1\), then \(C_{\le T}=C\) and \(C_{>T}=\varnothing\). Throughout the rest of this subsection, sums such as \(\sum_{v_t\in C_{\le T}}\) are over the global arrival times \(t\) whose requests belong to the indicated set.

We first bound the cost in the suffix \(C_{> T}\) using the following property of the pivot. Intuitively, the property says that the pivot \(v_T\) is as close to the center as a uniformly random remaining request, since the probability that request \(u\) makes a balanced opening is proportional to \(\beta_t(u)\) which is negatively correlated with \(r_u\).

\begin{claim}[Balanced pivot]
\label{clm:pivot}
Fix a round \(t\), condition on the complete history immediately before that round, and assume that no balanced opening from \(C\) has occurred before round \(t\). Let \(Y\) be the multiset of remaining copies of \(C\). Conditional on \(T=t\), the pivot \(p=v_t\) satisfies
\[
\E[r_p\mid \mathcal H_{t-1},T=t]
\le
\frac1{|Y|}\sum_{u\in Y}r_u,
\]
where \(\mathcal H_{t-1}\) denotes the pre-round history.
\end{claim}

\begin{proof}
If \(t=1\), then \(\beta_1(u)=1\) for every \(u\in Y\). Conditional on \(T=1\), the pivot is therefore uniform over \(Y\), and the claimed inequality holds with equality.

Conditional on \(\mathcal H_{t-1}\), every remaining request copy is equally likely to occupy time \(t\). If \(u\in Y\) arrives, the conditional probability of a balanced opening is \(\beta_t(u)\). Consequently,
\[
\Prb[p=u\mid\mathcal H_{t-1},T=t]
=
\frac{\beta_t(u)}{\sum_{x\in Y}\beta_t(x)}.
\]

For fixed \(D_t\) and \(q_t\), the quantity \(\beta_t(u)=g_t((D_t-r_u)_+)\) is nonincreasing in \(r_u\). Thus, \(r_p\) is biased toward the smaller values.

To see this formally, order the elements of \(Y\) so that
\[
r_1\le\cdots\le r_k
\qquad\text{and hence}\qquad
\beta_1\ge\cdots\ge\beta_k.
\]
Then
\[
\left(\sum_{i=1}^k\beta_i\right)
\left(\sum_{i=1}^k r_i\right)
-k\sum_{i=1}^k\beta_i r_i
=
\sum_{i<j}(\beta_i-\beta_j)(r_j-r_i)
\ge0.
\]
Dividing by \(k\sum_i\beta_i\) proves the claim. If \(\sum_i\beta_i=0\), then the event \(T=t\) has probability zero.
\end{proof}

Thus, the pivot is biased toward smaller cluster radii. This is exactly the direction needed for serving the remaining requests.

\begin{claim}[Weighted suffix]
\label{clm:weighted-suffix}
The expected cost incurred by the requests in \(C_{>T}\) satisfies
\[
\E[\ALG(C_{>T})]
\le
2\E\!\left[
\sum_{v_t\in C_{>T}}
(1+q_t)r_{v_t}
\right].
\]
\end{claim}

\begin{proof}
If \(T=n+1\), then \(C_{>T}=\varnothing\), and the claim is immediate. Fix \(t\le n\), condition on the pre-round history \(\mathcal H_{t-1}\) and on \(T=t\), and let \(Y\) be the remaining multiset of requests from \(C\) immediately before round \(t\). Put \(k=|Y|\), \(R_Y=\sum_{u\in Y}r_u\), and let \(p=v_t\) denote the pivot. Thus, \(C_{>T}=Y\setminus\{p\}\). If \(k=1\), then \(C_{>T}=\varnothing\), so assume \(k\ge2\).

Let \(S\subseteq\{t+1,\ldots,n\}\) be the set of future global times occupied by \(C_{>T}=Y\setminus\{p\}\). Conditional on \(\mathcal H_{t-1}\) and \(T=t\), the set \(S\) is independent of the pivot identity, and the remaining requests will be placed uniformly at random into $S$.

Write
\[
\bar q_S=\frac1{k-1}\sum_{s\in S}q_s,
\qquad
\mu=\E[r_p\mid\mathcal H_{t-1},T=t].
\]
By Claim~\ref{clm:pivot}, \(k\mu\le R_Y\).

The pivot remains open throughout the suffix. Therefore, every future \(u\in C_{>T}\) has an available connection of length at most \(d(u,p)\le r_u+r_p\). Thus, a request arriving at time \(s\) has expected service cost at most $h_s(d(v_s,p)) \le (1+q_s)(r_{v_s}+r_p)$.

Conditional on \(\mathcal H_{t-1}\), \(T=t\), \(S\), and the pivot \(p\), uniform assignment of the future identities to \(S\) therefore bounds the expected suffix cost by
\[
(1+\bar q_S)
\left(
\sum_{u\in Y\setminus\{p\}}r_u
+(k-1)r_p
\right)
=
(1+\bar q_S)\bigl(R_Y+(k-2)r_p\bigr).
\]
Taking expectation over the pivot and using \(k\mu\le R_Y\) gives
\[
(1+\bar q_S)\bigl(R_Y+(k-2)\mu\bigr)
\le
2(1+\bar q_S)(R_Y-\mu).
\]
On the other hand, conditional on \(\mathcal H_{t-1}\), \(T=t\), and \(S\), the expected weighted radius of the requests in \(C_{>T}\) is
\[
\E\!\left[
\sum_{v_s\in C_{>T}}
(1+q_s)r_{v_s}
\,\middle|\,
\mathcal H_{t-1},T=t,S
\right]
=
(1+\bar q_S)(R_Y-\mu).
\]
Therefore
\[
\E[\ALG(C_{>T})\mid\mathcal H_{t-1},T=t,S]
\le
2\E\!\left[
\sum_{v_s\in C_{>T}}
(1+q_s)r_{v_s}
\,\middle|\,
\mathcal H_{t-1},T=t,S
\right].
\]
Taking expectations over \(S\), the pre-round history, and \(T\) proves the claim.
\end{proof}

We now bound the prefix. We will bound the true cost by splitting an expression describing the cost that ``pretends'' the distance is $a_t(u)$. 

\begin{claim}[Prefix charging]
\label{clm:one-step}
If \(u=v_t\in C_{\le T}\), then, conditional on the history immediately before round \(t\) and on \(v_t=u\), its expected (opening and connection) cost is at most
\[
\beta_t(u)+(1-\beta_t(u))a_t(u)+2(1+q_t)r_u.
\]
\end{claim}

\begin{proof}
If \(t=1\), then \(h_1(d(u,F_0))=1\), \(\beta_1(u)=1\), and \((1-\beta_1(u))a_1(u)=0\), so the claimed inequality is immediate. Assume henceforth that \(t\ge2\).

Put \(d_t=d(u,F_{t-1})\), \(\beta_t=\beta_t(u)\), and \(a_t=a_t(u)\). Since \(h_t\) is \((1+q_t)\)-Lipschitz and \(d_t-a_t\le2r_u\),
\[
h_t(d_t)
\le
h_t(a_t)+2(1+q_t)r_u = \beta_t(u)+(1-\beta_t(u))a_t(u)+2(1+q_t)r_u. \qedhere
\]
\end{proof}

The balanced opening probabilities in \(C_{\le T}\) pay for at most one balanced opening:
\[
\E\!\left[
\sum_{v_t\in C_{\le T}}\beta_t(v_t)
\right]
=
\Prb[T\le n]
\le1.
\]

Denote by \(H_C\) the term that cannot be charged to the offline connection cost:
\[
H_C:=\E\!\left[
\sum_{v_t\in C_{\le T}}
(1-\beta_t(v_t))a_t(v_t)
\right].
\]

\begin{claim}
\label{clm:prefix}
The prefix cost satisfies
\[
\E[\ALG(C_{\le T})]
\le
1+H_C+2\E\!\left[
\sum_{v_t\in C_{\le T}}
(1+q_t)r_{v_t}
\right].
\]
\end{claim}

\begin{proof}
Sum Claim~\ref{clm:one-step} over the requests in \(C_{\le T}\), then take expectations. The \(\beta\)-terms contribute at most one in expectation, the $(1-\beta)a$-terms give \(H_C\), and the remaining terms give the stated radius charge.
\end{proof}

Combining Claims~\ref{clm:prefix} and~\ref{clm:weighted-suffix} gives our final decomposition lemma.

\begin{lemma}[Cluster decomposition]
\label{clm:decomposition}
For every optimal cluster \(C\),
\[
\E[\ALG(C)]\le 1+H_C+2(1+\bar q)R(C).
\]
\end{lemma}

\begin{proof}
The sets \(C_{\le T}\) and \(C_{>T}\) form a partition of \(C\). Summing Claim~\ref{clm:prefix} and Claim~\ref{clm:weighted-suffix} gives
\[
\E[\ALG(C)]
\le
1+H_C+
2\E\!\left[
\sum_{v_t\in C}
(1+q_t)r_{v_t}
\right] = 1+H_C+2(1+\bar q)R(C),
\]
where the last equality follows from the random-order property that each request is equally likely to occupy each time.
\end{proof}

\paragraph{The fixed-\texorpdfstring{\(q\)}{q} case.}
Before analyzing a varying sequence, we specialize the framework to \(q_t\equiv q\).

\begin{claim}
\label{lem:fixed-q-excess}
If \(q_t\equiv q>0\), then every optimal cluster \(C\) satisfies \(H_C\le1/q\).
\end{claim}

\begin{proof}
Fix \(v_t\in C_{\le T}\), and write \(a=a_t(v_t)\) and \(\beta=\beta_t(v_t)=g_t(a)\). Recall that \(H_C\) is defined using the failed-excess term \((1-g_t(a))a\).

We claim that \((1-\beta)a\le \beta/q\). If \(a<1\), then \(\beta=qa\), so \((1-\beta)a=(1-qa)a\le a=\beta/q\). If \(a\ge1\), then \(\beta=1\) and \((1-\beta)a=0\), so the same inequality is immediate. Thus \((1-\beta_t(v_t))a_t(v_t)\le \beta_t(v_t)/q\) for every \(v_t\in C_{\le T}\).

Summing over \(v_t\in C_{\le T}\) and taking expectations gives
\[
H_C
\le
\frac1q\,
\E\!\left[
\sum_{v_t\in C_{\le T}}\beta_t(v_t)
\right]
=
\frac{\Prb[T\le n]}{q}
\le
\frac1q . \qedhere
\]
\end{proof}

When \(q_t\equiv q\), we have \(\bar q=q\). Applying Lemma~\ref{clm:decomposition} and Claim~\ref{lem:fixed-q-excess} to an optimal cluster \(C\) gives
\[
\E[\ALG(C)]\le 1+\frac1q+2(1+q)R(C).
\]
Since \(\OPT(C)=1+R(C)\), summing over the optimal clusters proves the \(\max\{1+1/q,\,2(1+q)\}\) bound of~\cite{Kaplan23Almost}.

\subsection{Bounding \(H_C\) by a Relaxed Process}
\label{subsec:relaxation}

Lemma~\ref{clm:decomposition} reduces the analysis of a cluster \(C\) to the center-excess term \(H_C\). This is the only term that depends on the shape of the clock \(q_1,\ldots,q_n\), rather than just on its average. The difficulty is that the center excesses are adaptive: future openings may decrease \(D_t=d(c,F_{t-1})\), and hence may decrease all remaining excesses in the cluster. We upper-bound this adaptive metric process by a recursive excess process that allows the excesses to decrease arbitrarily.

\begin{definition}[Relaxed excess process]
We consider the following process. At the beginning of round \(t\), let \(k:=n-t+1\) requests remain. A state is a vector \(\mathbf x=(x_1,\ldots,x_k)\in[0,1]^k\) denoting the current excesses. The next request is chosen uniformly. If request \(i\) is revealed, then two things can happen:
\begin{enumerate}
    \item (Balanced opening) With probability \(q_t x_i\), the process stops.
    \item (Otherwise) With the remaining probability, we pay \(x_i\), remove request \(i\), and the adversary may lower the remaining excesses arbitrarily; that is, the process may continue from any vector in \([0,1]^{k-1}\) that is coordinatewise at most the remaining excess vector.
\end{enumerate}
The expected cost in this process is captured by the following recurrence. For \(\mathbf x\in[0,1]^k\), write \(\mathbf x_{-i}\in[0,1]^{k-1}\) for the vector obtained by removing the \(i\)th coordinate of \(\mathbf x\). We set \(U_{n+1}(\varnothing)=0\), and define
\[
    U_t(x_1,\ldots,x_k)
    :=
    \frac1{k}\sum_{i=1}^{k}
    (1-q_t x_i)
    \left(
        x_i+
        \sup_{\mathbf y \le \mathbf x_{-i}} U_{t+1}(\mathbf y)
    \right),
\]
where inequalities between vectors are coordinatewise.
\end{definition}

Consider requests with true excesses \(a_1,\ldots,a_k\), and let \(\hat a_i=\min\{a_i,1\}\). We will show that \(U_t(\hat{\mathbf a})\) is an upper bound on the actual excess we pay. The capping is harmless: if \(a_i\ge1\), then the actual failed-excess charge is \(0\). The relaxed process only uses the fact that, after processing a request, the remaining clipped excesses can only decrease.

\begin{lemma}[Relaxed-process domination]
\label{lem:relaxation}
For every optimal cluster \(C\),
\[
    H_C\le \sup_{\mathbf x \in [0,1]^n}U_1(\mathbf x).
\]
\end{lemma}

\begin{proof}
For an actual active state \(\mathcal S\) immediately before round \(t\), let \(W_t(\mathcal S)\) be the conditional expected remaining contribution to \(H_C\). For every unrevealed request \(i\), let \(a_i\) be its current center excess, with \(a_i=0\) for requests outside \(C\), and put \(\hat a_i=\min\{a_i,1\}\). We prove by backward induction that \(W_t(\mathcal S)\le U_t(\hat{\mathbf a})\).

The assertion is immediate after the last round. Fix an active state before round \(t\), and condition on request \(i\) being revealed. First, capping at \(1\) is harmless: \(1-g_t(a_i)\le 1-q_t\hat a_i\) and \((1-g_t(a_i))a_i\le (1-q_t\hat a_i)\hat a_i\). Indeed, if \(a_i<1\), then both inequalities hold with equality, while if \(a_i\ge1\), then \(1-g_t(a_i)=0\).

If the balanced attempt fails, request \(i\) is removed. Moreover, \(D_t=d(c,F_{t-1})\) cannot increase, so every remaining center excess, and therefore every remaining clipped excess, can only decrease. Suppose that the resulting clipped excess vector is \(\mathbf y\in[0,1]^{n-t}\). By the induction hypothesis, the continuation value is at most \(U_{t+1}(\mathbf y)\). Thus, the contribution conditional on revealing request \(i\) is at most
\[
    (1-g_t(a_i))a_i+(1-g_t(a_i))U_{t+1}(\mathbf y)
    \le
    (1-q_t\hat a_i)
    \left(\hat a_i+\sup_{\mathbf b \le \hat{\mathbf a}_{-i}} U_{t+1}(\mathbf b)\right).
\]
Averaging over the uniformly random next request gives exactly the recurrence for \(U_t(\hat{\mathbf a})\). This completes the induction, and applying the bound to the initial state proves the lemma.
\end{proof}

The metric dependence has now been removed. It remains only to understand how an adversary should choose the excess profile. We next show that the relaxed process is maximized when all remaining excesses are equal and this common level stays fixed throughout the process. This will immediately imply \cref{thm:general-q}.

For technical convenience, we work with an envelope of \(U_t\). Let \(k_t=n-t+1\). For \(0\le L\le k_t\), define the total-budget envelope
\[
V_t(L):=
\max\left\{
U_t(\mathbf x):
\mathbf x\in[0,1]^{k_t},
\sum_i x_i\le L
\right\}.
\]

We will show that this maximizer is achieved when all the excess values are equal and remain the same. To do so, we first show that in each round the maximum is attained by distributing the excess evenly.

\begin{claim}[One-round equalization]
\label{clm:one-round-equalization}
Let \(k_t=n-t+1\). For every \(t\in[n]\) and every \(L\in[0,k_t]\), the maximum defining \(V_t(L)\) is attained by assigning the same excess \(x\) to every remaining request. In particular,
\[
V_t(L)=
\max_{0\le x\le\min\{1,L/k_t\}}
(1-q_t x)\bigl(x+V_{t+1}((k_t-1)x)\bigr).
\]
Moreover, \(V_t(0)=0\), and \(V_t\) is nonnegative, nondecreasing, concave, and \(1\)-Lipschitz on \([0,k_t]\).
\end{claim}

\begin{proof}
We prove the claim by backward induction on \(t\). The properties are immediate for \(V_{n+1}=0\). Fix \(t\), put \(k=k_t\), and assume the statement holds for \(V_{t+1}\).

We first show that, for a fixed total excess, equalizing the current vector attains the maximum. Let \(\mathbf x=(x_1,\ldots,x_k)\in[0,1]^k\) have total \(S=\sum_i x_i\le L\). If request \(i\) is revealed, then after failure, the remaining total excess is at most \(S-x_i\). Hence
\[
U_t(\mathbf x)
\le
\frac1k\sum_{i=1}^k
(1-q_t x_i)\bigl(x_i+V_{t+1}(S-x_i)\bigr).
\]
For this fixed \(S\), define \(Y_s(x):=x+V_{t+1}(S-x)\). The function \(Y_s\) is nonnegative and concave. It is also nondecreasing: if \(x<y\), the \(1\)-Lipschitz property of \(V_{t+1}\) gives
\[
Y_s(y)-Y_s(x)
=(y-x)+V_{t+1}(S-y)-V_{t+1}(S-x)\ge0.
\]
Thus, as shown in \cref{clm:concave1}, the function $x \mapsto (1-q_tx)Y_s(x)$ is concave.\footnote{This is because it is the product of a nonnegative decreasing function ($1-q_t x$) and a nonnegative nondecreasing concave function $Y_s$. \cref{clm:concave1} formalizes this reasoning.} Jensen's inequality gives
\[
U_t(\mathbf x)
\le
\left(1-q_t\frac Sk\right)
\left(
\frac Sk+V_{t+1}\left(\frac{k-1}{k}S\right)
\right).
\]
Writing \(x=S/k\) and maximizing over \(S\le L\), we obtain
\[
V_t(L)
\le
\max_{0\le x\le\min\{1,L/k\}}
(1-q_t x)\bigl(x+V_{t+1}((k-1)x)\bigr) = \max_{0\le x\le\min\{1,L/k\}} G_t(x).
\]

The reverse inequality follows by taking the equal vector \((x,\ldots,x)\). After a failed reveal, the remaining vector is coordinatewise at most \(x\), and by the induction hypothesis the value \(V_{t+1}((k-1)x)\) is attained by an equal vector with each coordinate at most \(x\). Thus, this continuation is feasible in the coordinatewise relaxed process, giving equality.

It remains to propagate the regularity of \(V_t\). Let \[G_t(x):=(1-q_t x)(x+V_{t+1}((k-1)x)).\]
The same concavity argument gives that $G_t(x)$ is nonnegative and concave. It is also \(k\)-Lipschitz: for \(0\le x<y\le1\),
\[
G_t(y)-G_t(x)
\le
\bigl(y+V_{t+1}((k-1)y)\bigr)
-
\bigl(x+V_{t+1}((k-1)x)\bigr)
\le k(y-x).
\]
Recall \(V_t(L)=\max_{0\le x\le L/k}G_t(x)\). Since $G_t$ is concave, nonnegative, and $G_t(0) = 0$, its running maximum is obtained by following $G_t$ up to its first maximizer and is constant thereafter, and hence is concave. Furthermore, it is nondecreasing and $k$-Lipschitz as a function of $L/k$. Therefore \(V_t\) is nonnegative, nondecreasing, concave, and \(1\)-Lipschitz. Also \(V_t(0)=0\). This completes the induction.
\end{proof}

We now establish that for nonincreasing $q_t$, the maximum is achieved when the excess value does not decrease over time.

\begin{lemma}[Equalization]
\label{clm:adaptive-equalization}
For every \(L\in[0,n]\) and sequence $\{q_t\}$ such that $q_t \ge q_{t+1}$,
\[
V_1(L)=
\max_{0\le x\le \min\{1,L/n\}}
\sum_{t=1}^n x\prod_{s\le t}(1-q_sx).
\]
\end{lemma}

\begin{proof}
Define the value of a deterministic suffix profile recursively by \(\Phi_{n+1}=0\) and
\[
\Phi_t(x_t,\ldots,x_n)
:=
(1-q_t x_t)\bigl(x_t+\Phi_{t+1}(x_{t+1},\ldots,x_n)\bigr).
\]
Claim~\ref{clm:one-round-equalization} shows that \(V_1(L)\) is achieved by some nonincreasing profile \(1\ge x_1\ge x_2\ge\cdots\ge x_n\ge0\), with \(x_1\le L/n\), such that \(V_1(L)=\Phi_1(x_1,\ldots,x_n)\).

We next show that, among nonincreasing profiles with a fixed total excess, the constant profile is an upper bound. Fix a nonincreasing profile \(x_1,\ldots,x_n\). Consider adjacent positions \(j,j+1\), and let \(m=(x_j+x_{j+1})/2\). We show the following single-step averaging claim:
\begin{claim}
    $\Phi_j(x_j, x_{j+1}, x_{j+2},\ldots, x_n) \le \Phi_j(m, m, x_{j+2},\ldots, x_n)$.
\end{claim}
\begin{proof}
Write \(K=\Phi_{j+2}(x_{j+2},\ldots,x_n)\), \(a=q_j\), \(b=q_{j+1}\), \(z=x_j\), and \(w=x_{j+1}\). Thus \(a\ge b\), \(z\ge w\), and \(K\ge0\). Keeping the rest of the profile fixed,
\[
\Phi_j(z,w,x_{j+2},\ldots,x_n)
=
(1-a z)\bigl(z+(1-b w)(w+K)\bigr).
\]
Write $F(z, w) = (1-a z)\bigl(z+(1-b w)(w+K)\bigr)$. We want to compare $F(z,w)$ with $F(m,m)$.

    For $s\in[0,(z-w)/2]$, set $Z_s=z-s$ and $W_s=w+s$, and define
$\psi(s)=F(Z_s,W_s)$.  Along this path, $Z_s\ge W_s$.  It is enough to show
that $\psi$ is nondecreasing.  Differentiating gives
\[
\begin{aligned}
    \psi'(s)
    &=
    -\partial_Z F(Z_s,W_s)+\partial_W F(Z_s,W_s) \\
    &=
    (a-b)(K+2W_s)
    +a(Z_s-W_s)(1+bK+2bW_s)
    +abW_s^2 .
\end{aligned}
\]
Every term is nonnegative, since $a\ge b$, $Z_s\ge W_s$, and
$a,b,K,W_s\ge 0$.  Therefore $\psi'(s)\ge 0$ for the entire path, proving the claim.
\end{proof}
Consequently, replacing the adjacent pair $(x_j,x_{j+1})$ by $(m,m)$ cannot decrease $\Phi_j$. Repeatedly applying this averaging step and using compactness and continuity shows that a constant profile attains the maximum. Therefore, for some \(x\in[0,\min\{1,L/n\}]\),
\[
V_1(L)=\Phi_1(x,\ldots,x)
=
\sum_{t=1}^n x\prod_{s\le t}(1-q_sx).
\]
Taking the maximum over such \(x\) proves the lemma.
\end{proof}

We are now ready to prove \cref{thm:general-q}.

\begin{proof}[Proof of \cref{thm:general-q}]
Combining Lemmas~\ref{lem:relaxation} and~\ref{clm:adaptive-equalization} gives
\[
H_C\le \sup_{\mathbf x \in [0,1]^n}U_1(\mathbf x)
=V_1(n)
=
\sup_{x \in [0,1]}\sum_{t=1}^n x\prod_{s\le t}(1-q_sx)
=\rho(q).
\]
Together with Lemma~\ref{clm:decomposition}, this implies
\[
\E[\ALG(C)]
\le1+\rho(q)+2(1+\bar q)R(C).
\]
Since \(\OPT(C)=1+R(C)\), summing over all optimal clusters proves \cref{thm:general-q}.
\end{proof}
Thus all metric adaptivity has been compressed into the scalar functional \(\rho(q)\).

\subsection{Benchmark-Optimal Sequence}
\label{subsec:optimize}
Equipped with \Cref{thm:general-q}, we now compute the optimal factor that the algorithm and analysis can provide. We first show that the optimal sequence for this upper bound is a two-phase sequence. To do so, we show the claim for every fixed mean.

\begin{claim}[Discrete two-phase extremality]
\label{clm:discrete-two-phase-exact}
Fix \(M\in[0,n]\).  Among all nonincreasing sequences \(q_1,\ldots,q_n\in[0,1]\) satisfying
\(\sum_t q_t=M\), the functional $\rho(q)$ is minimized by the front-loaded sequence
\[
q_t^*:=
\begin{cases}
1,&t\le m,\\
\theta,&t=m+1,\\
0,&t\ge m+2,
\end{cases}
\]
where \(m:=\lfloor M\rfloor\) and \(\theta:=M-m\).
\end{claim}

\begin{proof}
Let $\Phi_x := \sum_{t=1}^n x \prod_{s \le t}(1-q_s x)$. It is enough to prove that $\Phi_x$ is minimized by $q^*$ for every fixed \(x\), since we may then take the supremum over \(x\).

Suppose the nonincreasing sequence is not \(q^*\). Then pick the first index $i$ with $q_i < 1$ and the last index $j$ with $q_j > 0$. Since the sequence is nonincreasing, $i < j$. Move an amount $\varepsilon = \min\{1-q_i,q_j\}$ from \(q_j\) to \(q_i\), leaving all other coordinates unchanged. The resulting sequence remains nonincreasing, and we show that this weakly decreases \(\Phi_x\).

For prefixes ending before \(i\), nothing changes.  For prefixes ending at
some \(t\) with \(i\le t<j\), the prefix product is multiplied by replacing
the factor \(1-q_ix\) with \(1-(q_i+\varepsilon)x\), so it can only decrease.

For prefixes ending at \(t\ge j\), only the product of the two changed factors
matters.  Since \(q_i\ge q_j\), we have
\[
\begin{aligned}
\bigl(1-(q_i+\varepsilon)x\bigr)
 \bigl(1-(q_j-\varepsilon)x\bigr)
-
(1-q_ix)(1-q_jx) =
-x^2\varepsilon(q_i-q_j+\varepsilon)
\le0.
\end{aligned}
\]
Thus, no prefix product increases and $\Phi_x$ cannot increase.

Repeating this left-shift operation moves all mass as far left as possible and gives \(\Phi_x(q^*)\le\Phi_x(q)\) for every \(x\), and
taking the supremum over \(x\) gives \(\rho(q^*)\le\rho(q)\).
\end{proof}

We now compute the optimal asymptotic choice of $M$ and prove \cref{thm:sequence}. It suffices to analyze the integral front-loaded sequence. If $M$ is not an integer, rounding $M$ up to $\lceil M \rceil$ only increases $\bar q$ by at most $1/n$ and can only decrease $\rho(q)$.

\begin{proof}[Proof of \cref{thm:sequence}]
For such two-phase sequences, the total value becomes
\[
\rho_{n,m}
:=
\sup_{0\le x\le1}
\left[
(1-x)\bigl(1-(1-x)^m\bigr)
+
(n-m)x(1-x)^m
\right].
\]

Let \(\alpha=m/n\).  We compare this finite expression to the continuum
function
\[
F_\alpha(L):=1-e^{-\alpha L}+(1-\alpha)Le^{-\alpha L}.
\]

\begin{claim}
    $\rho_{n,m}\le \sup_{L\ge0}F_{\alpha}(L)$.
\end{claim}
\begin{proof}
Fix \(x\in[0,1]\), set \(L=nx\), and write \(y=1-x\).  The second term
satisfies
\[
(n-m)x(1-x)^m
=
(1-\alpha)L(1-x)^m
\le
(1-\alpha)L e^{-\alpha L}.
\]
For the first term, use the identity
\[
(1-x)\bigl(1-(1-x)^m\bigr)
=
\int_0^x m(1-x)(1-s)^{m-1}\,ds.
\]
For \(0\le s\le x\), we have \(1-x\le1-s\), and hence $(1-x)(1-s)^{m-1}\le(1-s)^m\le e^{-ms}$.
Therefore
\[
(1-x)\bigl(1-(1-x)^m\bigr)
\le
\int_0^x m e^{-ms}\,ds
=
1-e^{-mx}
=
1-e^{-\alpha L}.
\]
Adding the two bounds gives
\[
(1-x)\bigl(1-(1-x)^m\bigr)
+
(n-m)x(1-x)^m
\le
F_{\alpha}(L).
\]
Taking the supremum over \(x\in[0,1]\) proves the claim.
\end{proof}

Now compute the continuum supremum.  For fixed \(\alpha\in(0,1)\), $F_\alpha'(L)=e^{-\alpha L}\bigl(1-\alpha(1-\alpha)L\bigr)$, so the unique maximizer is \(L^*=1/(\alpha(1-\alpha))\).  Hence
\[
\sup_{L\ge0}F_\alpha(L)
=
1+
\frac{1-\alpha}{\alpha}
\exp\!\left(-\frac1{1-\alpha}\right).
\]

Thus, for every finite \(n\), let $\alpha = m/n$,
\[
\rho_{n,m}
\le
1+
\frac{1-\alpha}{\alpha}
\exp\!\left(-\frac1{1-\alpha}\right).\]

The competitive ratio is then at most \[\max\left \{2 + 
\frac{1-\alpha}{\alpha}
\exp\!\left(-\frac1{1-\alpha}\right), 2(1+\alpha)\right \}.\]

Let \(\alpha^* \approx 0.293\) be the unique solution of
\[
\frac{1-\alpha}{\alpha}
\exp\!\left(-\frac1{1-\alpha}\right)
=
 2\alpha.
\]

Since the left-hand side is non-increasing on $(0, 1)$, we have for every $\alpha > \alpha^*$, the competitive ratio is at most
\begin{align*}
    2+\max\left \{
\frac{1-\alpha}{\alpha}
\exp\!\left(-\frac1{1-\alpha}\right), 2\alpha\right \} \le 2+\max\left \{
\frac{1-\alpha^*}{\alpha^*}
\exp\!\left(-\frac1{1-\alpha^*}\right), 2\alpha\right \} = 2(1+\alpha).
\end{align*}
Taking $m = \lceil \alpha^* n\rceil$ achieves $2(1+\alpha^*)+O(1/n)$.
\end{proof}

\subsection{Adversarial-Order Robustness}
\label{subsec:robust}
Finally, we prove \cref{thm:bounded-clock-adversarial}, showing that any \(q_t\)-DistProb algorithm whose multipliers are bounded above and below by positive constants retains the asymptotically optimal adversarial-order guarantee. The proof is essentially the same as the original analysis of \DistProb in~\cite{Fotakis2003CompetitiveRatio, Meyerson01online}, so we provide only a sketch on the differences here and defer the complete proof to Appendix~\ref{appendix:robust}.

\begin{proof}[Proof Sketch of \cref{thm:bounded-clock-adversarial}]
For every $x$, let $\hat x:= \min\{x, 1\}$.
For every possible history and every current distance \(d=d(v_t,F_{t-1})\), the opening probability \[g_t(d)=\begin{cases}
    1 & d \ge 1\\
    q_t d & d < 1
\end{cases}\] satisfies $q_- \hat d \le g_t(d) \le \max(1,q_+) \hat d$.

The adversarial-order analysis of \DistProb in~\cite{Fotakis2003CompetitiveRatio, Meyerson01online} applies to any algorithm whose opening probability is within constant factors of \(\hat d\).  We recall the two places where the opening probability is used.
\begin{itemize}
    \item (The connection cost) The analysis considers requests assigned
to an optimal cluster before a nearby facility has been opened.  If the current
connection distance is \(d\), the probability of not opening is at most $1-q_-\hat d\le \exp(-q_-\hat d)$. Thus, the usual waiting-time estimates increase by at most a factor \(1/q_-\).
    \item (The opening cost) The expected number of facilities opened in each scale is bounded by the sum of the opening probabilities. Since $g_t(d) \le \max(1,q_+) \hat d$, the opening-cost part of the original proof increases by at most a factor \(\max(1,q_+)\).
\end{itemize}

All other parts of the charging argument are deterministic metric inequalities and are unchanged.  Therefore, the same scale decomposition for \DistProb gives
\[
    O_{q_-,{q_+}}\!\left(\frac{\log n}{\log\log n}\right)
\]
for the bounded $q_t$-\DistProb algorithm.
\end{proof}
 \section{Lower Bounds}
\label{sec:lowerbound}
In this section, we complement our upper bounds with lower bounds for the algorithmic families considered in this paper. All families will only open facilities only at request locations, which we call \emph{request-local}. We show a lower bound of $3$ for time-oblivious request-local rules (\cref{thm:memoryless}), a lower bound of $\approx 2.519$ for the $q_t$-\DistProb family (\cref{thm:qt-family-lower-bound}), and a lower bound of $\approx 2.42$ for the \TimeDist family.

All lower bounds come from the same sparse-versus-dense tension. A request-local algorithm should avoid opening in sparse regions, because an offline solution may open at the center and serve many leaves cheaply. On the other hand, it should open quickly in dense regions, because repeated requests at the same or nearby locations otherwise incur repeated connection cost. The algorithmic classes considered below differ mainly in what information they can use to distinguish these two cases.

\paragraph{Two recurring templates.}
The first template is a sparse star. For a parameter \(\delta>0\), take a star metric with a center \(c\) and leaves \(u_1,\ldots,u_m\), with \(d(c,u_i)=\delta\) and \(d(u_i,u_j)=2\delta\) for distinct leaves. The sparse-star instance has one request at each leaf. The offline solution opens \(c\), while a request-local algorithm can open only at the leaves. This instance penalizes unnecessary openings at fresh leaves.

The second template is a dense-location instance. Here a location, or a prescribed sequence of locations, receives many request copies. Until the algorithm opens at such a location, every copy pays essentially the same connection cost. This instance penalizes algorithms that wait too long before opening.

\subsection{Time-oblivious Request-Local Algorithms}
\label{subsec:memoryless}
We first show that time-oblivious request-local rules cannot beat \(3\). We say that an algorithm is a time-oblivious request-local algorithm if it opens facilities only at request locations and its decision at time $t$ depends only on the metric space, the current request $v_t$, and the current open facility set $F_{t-1}$.\footnote{Such algorithms are sometimes also called \emph{memoryless}~\cite{fotakis11memoryless}; here we call them time-oblivious to emphasize that they do not use time when making the decisions.}

\begin{theorem}[Time-oblivious lower bound]
\label{thm:memoryless}
Every randomized time-oblivious request-local rule for random-order online facility location has competitive ratio at least \(3\).
\end{theorem}

\begin{proof}
Fix such an algorithm and suppose, for contradiction, that it is \((3-\varepsilon)\)-competitive for some \(0<\varepsilon<1\). Set \(\gamma=\varepsilon/4\), and choose \(0<\delta<\min\{1/4,\varepsilon/16\}\). For each \(m\), consider the star metric with center \(c\), labeled leaves \(u_1,\ldots,u_m\), and \(d(c,u)=\delta\), \(d(u,v)=2\delta\) for distinct leaves \(u,v\). For a set \(F\) of open leaves and a fresh leaf \(v\notin F\), let \(p(F,v)\) be the algorithm's conditional probability of opening \(v\).

We first derive a constraint along every labeled opening chain. Let \(x_1,\ldots,x_r\) be any ordered tuple of distinct leaves, with \(r\le m\), and put \(F_i=\{x_1,\ldots,x_i\}\). For an integer \(N\ge2\), create an instance with \(N^{r-i+1}\) copies of \(x_i\). Thus, the multiplicities decrease geometrically along the prescribed ordering. Fix \(B\). For each \(i<r\), consider the projection of the random sequence onto the copies of \(x_i,\ldots,x_r\), and let \(E_{N,B}\) be the event that the first \(B\) requests in this projection are all at \(x_i\), simultaneously for all \(i<r\). Since the total multiplicity of later leaves is smaller than the multiplicity of \(x_i\) by a factor \(O_r(1/N)\), we have \(\Pr[E_{N,B}]=1-O_{r,B}(1/N)\).

On \(E_{N,B}\), after ignoring requests at leaves that are already open, the sequence offers at least \(B\) attempts at \(x_1\) before any later leaf, then at least \(B\) attempts at \(x_2\) before any still later leaf, and so on. The first request must be opened. While the open set is \(F_i\), every attempt at \(x_{i+1}\) has conditional success probability \(p(F_i,x_{i+1})\), regardless of preceding failures. Let \(C_B\) be the cost of the truncated chain experiment in which each stage is given at most \(B\) attempts and the experiment stops if a stage fails all of them. The actual cost dominates \(C_B\) on \(E_{N,B}\), and hence \(\liminf_{N\to\infty}\E[\ALG]\ge \E[C_B]\). Letting \(B\to\infty\), monotone convergence gives
\[
\lim_{B\to\infty}\E[C_B]
=
r+2\delta\sum_{i=1}^{r-1}\left(\frac1{p(F_i,x_{i+1})}-1\right).
\]
Opening one facility at every active leaf costs \(r\), so \(\OPT\le r\). The assumed competitive ratio therefore implies
\[
\sum_{i=1}^{r-1}\frac1{p(F_i,x_{i+1})}
\le
\left(\frac{2-\varepsilon}{2\delta}+1\right)r-1.
\]

For every nonempty set \(F\) of leaves, define the path potential
\[
H(F)=
\max_{(x_1,\ldots,x_s)\in\operatorname{Perm}(F)}
\sum_{i=1}^{s-1}\frac1{p(\{x_1,\ldots,x_i\},x_{i+1})},
\qquad s=|F|.
\]
The chain bound gives, for every nonempty \(F\),
\[
H(F)\le
\left(\frac{2-\varepsilon}{2\delta}+1\right)|F|-1
\le
\frac{1-\gamma}{\delta}|F|,
\]
where the last inequality follows from \(1-\varepsilon/2+\delta\le1-\gamma\).

We now use the sparse-star instance with one request at each of the \(m\) leaves. Let \(F_t\) be the set of open leaves after \(t\) requests and let \(S_m=|F_m|\). The first request is necessarily opened, so \(H(F_1)=0\). For every later request at a fresh leaf \(v_t\notin F_{t-1}\), conditional on the complete history and on \(v_t\), the algorithm opens with probability \(p(F_{t-1},v_t)\). Hence
\[
\E\!\left[H(F_t)-H(F_{t-1})\,\middle|\,\mathcal F_{t-1},v_t\right]
\ge p(F_{t-1},v_t)\cdot \frac1{p(F_{t-1},v_t)}=1.
\]
Summing over \(t=2,\ldots,m\), we get \(\E[H(F_m)]\ge m-1\). Combining this with the pointwise bound on \(H\), we obtain \(\E[S_m]\ge \delta(m-1)/(1-\gamma)\).

Each request in the sparse-star instance costs \(1\) if the algorithm opens and \(2\delta\) otherwise. Therefore \(\E[\ALG]=2\delta m+(1-2\delta)\E[S_m]\). The offline optimum opens \(c\) and pays at most \(1+m\delta\). Thus
\[
\liminf_{m\to\infty}\frac{\E[\ALG]}{\OPT}
\ge
2+\frac{1-2\delta}{1-\gamma}>3,
\]
where the strict inequality follows from \(2\delta<\gamma\). This contradicts \((3-\varepsilon)\)-competitiveness.
\end{proof}

\subsection{$q_t$-\(\DistProb\) Family}
\label{subsec:qt-lower-bound}

We next consider the \(q_t\)-\DistProb family. At round \(t\), if the current connection distance is \(d<1\), the algorithm opens the request with probability \(q_t d\), where \(q_t\in[0,1]\); if \(d\ge1\), it opens deterministically. The values \(q_t\) may depend arbitrarily on the horizon and on the time \(t\). Let \(\bar q_n=n^{-1}\sum_{t=1}^n q_t\).

\begin{theorem}[$q_t$-\DistProb lower bound]
\label{thm:qt-family-lower-bound}
For every $\{q_t\}$ such that $q_t \in [0, 1]$ for all $t$, the algorithm $q_t$-\(\DistProb\) has asymptotic competitive ratio at least $2(1+\mu^*) \ge 2.519$, where $\mu^*$ is the solution to
\[
2(1+\mu)=
2+\frac{1-\mu}{\mu}\exp\left(-\frac{1+\mu}{1-\mu}\right)
.
\]
\end{theorem}

\begin{proof}
We show that if $\bar q_n \to \mu$, then the algorithm has competitive ratio at least 
\[
\max\left\{
2(1+\mu),\;
2+\frac{1-\mu}{\mu}\exp\left(-\frac{1+\mu}{1-\mu}\right)
\right\}.
\]

The sparse-star branch uses a star with center \(c\), leaves \(u_1,\ldots,u_n\), and distances \(d(c,u_i)=\delta/2\), \(d(u_i,u_j)=\delta\). There is one request at each leaf. The first request opens. Every later request is at a fresh leaf and has distance \(\delta<1\) from every previously opened leaf, so the expected cost at time \(t\) is \(q_t\delta+(1-q_t\delta)\delta=\delta(1+q_t)-q_t\delta^2\). The offline optimum opens \(c\) and pays at most \(1+n\delta/2\). Taking \(\delta\to0\), \(n\delta\to\infty\), and \(\bar q_n\to\mu\), we obtain the lower bound \(2(1+\mu)\).

For the dense branch, take \(m\) locations at pairwise distance \(d<1\), and put \(K\) request copies at each location. Let \(n=mK\), and take \(m\to\infty\), \(d\to0\), and \(Kd\to L\). The offline optimum opens every location, so \(\OPT\le m\). Fix one location. Until it opens, every request there has connection distance \(d\), and a request at time \(t\) opens with probability \(q_t d\). Define \(Q_n(z)=n^{-1}\sum_{t\le zn}q_t\). In the scaling limit, the probability that the location has not opened before normalized time \(z\) is at least \(\exp(-LQ_n(z))-o(1)\), and the eventual opening probability is \(1-\exp(-LQ_n(1))+o(1)\). Hence the expected cost contributed by this location is at least
\[
1-\exp(-LQ_n(1))+\int_0^1 L\exp(-LQ_n(z))\,dz-o(1).
\]

Since \(0\le q_t\le1\), we have \(Q_n(z)\le z\) and \(Q_n(z)\le Q_n(1)=\bar q_n\). Along a subsequence with \(\bar q_n\to\mu\), this gives \(Q_n(z)\le\min\{z,\mu\}+o(1)\), and therefore the repeated-location instance gives the lower bound
\[
1-e^{-\mu L}+\int_0^1Le^{-L\min\{z,\mu\}}\,dz
=
2(1-e^{-\mu L})+(1-\mu)Le^{-\mu L}.
\]
Maximizing over \(L\ge0\), the derivative vanishes at \(L^\star=(1+\mu)/(\mu(1-\mu))\), and the maximum value is
\[
2+\frac{1-\mu}{\mu}\exp\left(-\frac{1+\mu}{1-\mu}\right).
\]

Combining the sparse and dense branches, every subsequential limit \(\mu\) of \(\bar q_n\) yields the stated lower bound. The first term is increasing in \(\mu\), while the second is decreasing. Their unique intersection satisfies
\[
2\mu=\frac{1-\mu}{\mu}\exp\left(-\frac{1+\mu}{1-\mu}\right).
\]
The solution is \(\mu=0.259\ldots\), giving value \(2(1+\mu)=2.519\ldots\).
\end{proof}

\subsection{Time--Distance Rules}
\label{sec:time-distance-lower-bound}
We now prove optimality of our $2.42$ ratio among the broader class of \TimeDist algorithms. Recall that such an algorithm opens a facility at the current request $v_t$ with probability \(g(t,d(v_t, F))\) for some fixed function $g$.

Let \(\mu^\star \approx 0.21\) be the unique positive solution of \(2(1+\mu)=1+e^{-(1+\mu)}/\mu\), and set \(C^\star=2(1+\mu^\star)\).

\begin{theorem}[\TimeDist lower bound]
\label{thm:time-distance-lower-bound}
Every algorithm in the \TimeDist family has asymptotic competitive ratio at least \(C^\star\).
\end{theorem}

\begin{proof}
For every sufficiently large integer \(K\), we construct one of two hard instances with the same horizon and distance scale. Set 
\[
L^\star:=1+1/\mu^\star, \quad n:=K^3, \quad d:=L^\star/K.
\]
Thus \(d\to0\), \(nd\to\infty\), and \(Kd=L^\star\). Since the first request must be served when no facility is open, feasibility forces it to open. Define \(\widehat p_1:=1\), \(\widehat p_t:=g(t, d)\) for \(t\ge2\), \(\overline p:=n^{-1}\sum_t\widehat p_t\), and \(\mu_K:=\overline p/d\).

If \(\mu_K\ge\mu^\star\), use the sparse-star instance with \(n\) leaves, \(d(c,u_i)=d/2\), and \(d(u_i,u_j)=d\). There is one request at each leaf. The first request opens, and every later request is at distance \(d\) from the open leaves. The expected cost at time \(t\) is \(\widehat p_t+(1-\widehat p_t)d=d+(1-d)\widehat p_t\). Hence \(\E[\ALG]=n(d+(1-d)\overline p)=nd(1+(1-d)\mu_K)\). Since \(\OPT\le1+nd/2\), we get \(\E[\ALG]/\OPT\ge 2(1+\mu^\star)-o_K(1)=C^\star-o_K(1)\).

It remains to handle \(\mu_K\le\mu^\star\). Take \(M=K^2\) locations at pairwise distance \(d\), and place \(K\) copies at each location, so \(MK=K^3=n\). The offline optimum opens every location, so \(\OPT\le M\). Fix one location \(u_j\), and let \(S_j\subseteq[n]\) be the random set of times occupied by its \(K\) copies. Until \(u_j\) opens, every request there has distance \(d\). Conditional on \(S_j\), the probability that no facility is ever opened at \(u_j\) is \(\prod_{t\in S_j}(1-\widehat p_t)\), by the chain rule applied to the conditional failure probabilities.

Let \(P_0\) be this no-opening probability. To lower-bound it, sample \(T_1,\ldots,T_K\) independently and uniformly from \([n]\), and let \(D\) be the event that they are all distinct. Conditional on \(D\), the unordered set is a uniformly random \(K\)-subset. Since all factors lie in \([0,1]\),
\[
(1-\overline p)^K
=
\E\left[\prod_{i=1}^K(1-\widehat p_{T_i})\right]
\le
P_0+\Pr[D^c].
\]
The union bound gives \(\Pr[D^c]\le K(K-1)/(2n)\), and hence \(P_0\ge(1-\overline p)^K-K(K-1)/(2n)\). Since \(\overline p=d\mu_K\le d\mu^\star\), \(Kd=L^\star\), and \(K^2/n=1/K\), we get \(P_0\ge e^{-\mu^\star L^\star}-o_K(1)\).

If \(u_j\) is ever opened, the algorithm pays at least one unit there. If it is never opened, all \(K\) requests there connect at distance \(d\), for cost \(Kd=L^\star\). Thus, the expected cost charged to \(u_j\) is at least \(1+(L^\star-1)P_0\ge 1+(L^\star-1)e^{-\mu^\star L^\star}-o_K(1)\). Summing over the \(M\) locations and using \(\OPT\le M\), we obtain
\[
\frac{\E[\ALG]}{\OPT}
\ge
1+(L^\star-1)e^{-\mu^\star L^\star}-o_K(1)
=
1+\frac{e^{-(1+\mu^\star)}}{\mu^\star}-o_K(1)
=
C^\star-o_K(1).
\]
For every \(K\), either the sparse-star case or the repeated-location case applies. Therefore, along the subsequence \(n=K^3\), every time--distance rule has a hard instance with ratio at least \(C^\star-o_K(1)\). This proves the theorem.
\end{proof} 

\paragraph{Acknowledgements.} The authors would like to thank the participants at Dagstuhl Seminar 26131 ``New Trends in Clustering'' for inspiring discussions about the problem and thank Michael Mitzenmacher for insightful comments on earlier versions of the paper. ChatGPT 5.5-Pro was used for brainstorming ideas, polishing the exposition, and proofreading the manuscript. 

\printbibliography
\appendix
\section{An $\Omega(\sqrt n)$ Lower Bound for \DistCut in Adversarial Order}
\label{appendix:cut-adv}
\begin{claim}
\label{clm:distcut-not-adversarial}
For every fixed \(\mu>0\), \DistCut is not \(O(\log n/\log\log n)\)-competitive in the adversarial-order model. In fact, its competitive ratio can be \(\Omega(\sqrt n)\).
\end{claim}

\begin{proof}
Intuitively, an adversarial-order sequence can pad a given sequence with many additional requests at the end. Thus, the ``real'' sequence appears at the beginning of the padded sequence, and the algorithm opens too many facilities during this initial stage. We now formalize this idea.

Recall that \DistCut opens the current request \(v_t\) if $d(v_t, F_{t-1}) \ge \min\{1, (t-1)/(\mu n)\}$.
We construct an adversarial-order instance. Let $\lambda=n^{-1/2}$. Consider a uniform metric such that \(d(u,v)=\lambda\) for $u \ne v$. Let $m = \lfloor \mu\lambda n\rfloor+1$.
The adversary first presents the distinct points $u_1, \ldots, u_m$ and then presents \(n-m\) additional copies of \(u_1\).

The first request opens automatically. For every \(2\le t\le m\), the request \(v_t\) is fresh, and its distance to the previously opened leaves is exactly \(\lambda<1\). Moreover,
\[
    \frac{t-1}n
    \le
    \frac{m-1}n
    \le
    \mu\lambda.
\]
Hence \DistCut opens each of the first \(m\) leaves. Therefore $\ALG\ge m = \Theta(\mu \sqrt n)$.

On the other hand, the offline optimum can simply open \(u_1\), serve all copies of \(u_1\) at cost \(0\), and serve the other requests at cost \(\lambda\). Thus $\OPT \le 1+{(m-1)\lambda} = O(1)$.

We conclude that $\ALG \ge \Omega(\sqrt n) \OPT$.
\end{proof}

\section{Concavity}
Here we include the omitted proof of the concavity fact used in \cref{subsec:relaxation}.
\begin{claim}
\label{clm:concave1}
    Let \(I\) be an interval, let \(Y:I\to\mathbb R_{\ge0}\) be concave and
nondecreasing, and let $q \ge 0$. Suppose \(1-qx\ge0\) for every \(x\in I\). Then
\(x\mapsto(1-qx)Y(x)\) is concave on \(I\).
\end{claim}
\begin{proof}
Fix \(x\le y\), let \(\lambda\in[0,1]\), and put
\(z:=\lambda x+(1-\lambda)y\). Write \(\ell(s)=1-qs\). By concavity of
\(Y\) and affinity of \(\ell\),
\[
\ell(z)Y(z)\ge
\bigl(\lambda\ell(x)+(1-\lambda)\ell(y)\bigr)
\bigl(\lambda Y(x)+(1-\lambda)Y(y)\bigr).
\]
The difference between the right-hand side and
\(\lambda\ell(x)Y(x)+(1-\lambda)\ell(y)Y(y)\) is
\[
\lambda(1-\lambda)
\bigl(\ell(x)-\ell(y)\bigr)
\bigl(Y(y)-Y(x)\bigr)\ge0.
\]
Hence \(x\mapsto(1-qx)Y(x)\) is concave.
\end{proof}

\section{Robustness of $q_t$-\DistProb in Adversarial Order}
\label{appendix:robust}
Here, we prove \cref{thm:bounded-clock-adversarial} by adapting the standard proof of~\cite{Meyerson01online,Fotakis2003CompetitiveRatio}.
\begin{proof}[Proof of \cref{thm:bounded-clock-adversarial}]
As usual, we analyze one offline cluster $C$ with center $c$. Write \(d_t=d(F_{t-1},v_t)\). Recall that the algorithm opens \(v_t\) with probability
\[
    g_t(d_t)=
    \begin{cases}
        1, & d_t\ge 1,\\
        q_t d_t, & d_t<1,
    \end{cases}
\]
where \(0<q_-\le q_t\le q_+ \le 1\) for all \(t\). Let \(I_t\) be the indicator
that \(v_t\) opens a facility. We use two elementary inequalities. First,
conditional on the past,
\[
    \mathbb E[I_t+(1-I_t)d_t]\le (1+q_+)d_t .
    \tag{1}
\]
Indeed, if \(d_t\ge 1\), the left-hand side is \(1\le d_t\); and if \(d_t<1\),
it is \(q_t d_t+(1-q_t d_t)d_t\le (1+q_+)d_t\). Second,
\[
    (1-g_t(d_t))d_t\le \frac{1}{q_-}g_t(d_t).
    \tag{2}
\]
For \(d_t\ge 1\), the left-hand side is zero, and for \(d_t<1\), it is at most
\(d_t\le g_t(d_t)/q_-\).

Assume first that \(R(C)>0\), and put \(\delta=R(C)/|C|\). Fix integers
\(m,h\) with \(m^h>|C|\). We divide the requests of \(C\) into phases
\(h,h-1,\ldots,0\), followed by a final phase, as in Meyerson's analysis.
Phase \(j\) ends as soon as some facility is opened within distance
\(m^j\delta\) of \(c\). Thus, during phase \(j<h\), there is already an open
facility within distance \(m^{j+1}\delta\) of \(c\). A request \(u\in C\)
arriving in phase \(j\) is inner if \(r_u\le m^j\delta\), and outer otherwise.
All requests arriving in the final phase are treated as outer. Since
\(m^h\delta>|C|\delta=R(C)\), every request of \(C\) in phase \(h\) is inner.

We first bound the outer cost. If \(u\) is outer in a non-final phase \(j<h\),
then \(r_u>m^j\delta\), while some open facility is within distance
\(m^{j+1}\delta\) of \(c\). Therefore
\(d_t\le m^{j+1}\delta+r_u\le (m+1)r_u\). By (1), the expected cost charged to
\(u\) is at most \((1+q_+)(m+1)r_u\). If \(u\) arrives in the final phase, then
some open facility is within distance \(\delta\) of \(c\), so
\(d_t\le \delta+r_u\), and the expected cost charged to \(u\) is at most
\((1+q_+)(\delta+r_u)\). Hence the total expected outer cost charged to \(C\)
is at most
\[
    (1+q_+)(m+2)R(C).
    \tag{3}
\]

It remains to bound the inner cost. Fix a phase \(j\). Let \(E_t^j\) be the
event that \(v_t\in C\) is inner in phase \(j\), and that no earlier inner
request of \(C\) in the same phase has opened a facility. Using (2),
\[
\begin{aligned}
    \mathbb E\!\left[\sum_t \mathbf 1_{E_t^j}(1-I_t)d_t\right]
    &=
    \mathbb E\!\left[\sum_t \mathbf 1_{E_t^j}(1-g_t(d_t))d_t\right]  \\
    &\le
    \frac1{q_-}\mathbb E\!\left[\sum_t \mathbf 1_{E_t^j}g_t(d_t)\right] \\
    &=
    \frac1{q_-}\mathbb E\!\left[\sum_t \mathbf 1_{E_t^j}I_t\right]
    \le \frac1{q_-}.
\end{aligned}
\]
The expected facility-opening cost of the first inner request of \(C\) in this
phase is at most \(1\). Thus, the expected inner cost in phase \(j\) is at most
\(1+1/q_-\). There are \(h+1\) non-final phases, so the total expected inner
cost charged to \(C\) is at most
\[
    \left(1+\frac1{q_-}\right)(h+1).
    \tag{4}
\]

Combining (3) and (4), we get
\[
    \mathbb E[\ALG(C)]
    \le
    \left(1+\frac1{q_-}\right)(h+1)
    +
    (1+q_+)(m+2)R(C).
\]
Equivalently, for absolute constants \(A_q,B_q\) depending only on the clock
bounds,
\[
    \mathbb E[\ALG(C)]
    \le
    A_q h+B_q m\,R(C).
\]

If \(R(C)=0\), then all requests in \(C\) are located at \(c\). Applying the same
stopped-set argument once to all requests of \(C\), the expected cost until the
first such request opens is at most \(1+1/q_-\), and after that the remaining
requests have zero connection cost. Thus, the same bound holds.

Taking $m = h = \Theta(\log n/\log \log n)$ and summing over all offline clusters gives
\[
\E[\ALG] \le \max\left\{A_q, B_q\right\}O\!\left( \frac{\log n}{\log \log n}\right) \OPT  = O_{q_-, q_+}\left( \frac{\log n}{\log \log n}\right) \OPT. \qedhere
\]
\end{proof} 
\end{document}